\documentclass[11pt]{article}

\usepackage{fullpage}
\usepackage[utf8]{inputenc}
\usepackage{xcolor}
\usepackage{amsmath}
\usepackage{amssymb}
\usepackage{amsthm}
\usepackage{booktabs}
\usepackage[utf8]{inputenc}
\usepackage{csquotes}
\usepackage{bbm}
\usepackage[ruled,vlined]{algorithm2e}
\SetAlFnt{\small}
\usepackage{hyperref}
\usepackage{tikz}
\usetikzlibrary{calc, decorations.pathreplacing, calligraphy}
\usepackage{mathdots}
\usepackage{enumitem}
\allowdisplaybreaks

\usepackage{thmtools, thm-restate}
\declaretheorem{theorem}

\declaretheorem[sibling=theorem]{claim}

\newcommand{\mc}{\mathcal}
\newcommand{\A}{\mc A}
\newcommand{\B}{\mc B}
\newcommand{\I}{\mc I}
\newcommand{\T}{\mc T}
\newcommand{\U}{\mc U}

\newcommand{\OPT}{\ensuremath{\mathrm{OPT}}}
\newcommand{\ALG}{\ensuremath{\mathrm{ALG}}}

\newcommand{\fettered}{tight}

\title{Caching with Reserves}
\author{Sharat Ibrahimpur \\ University of Waterloo \\ sharat.ibrahimpur@uwaterloo.ca \and Manish Purohit \\ Google Research \\ mpurohit@google.com \and Zoya Svitkina \\ Google Research \\ zoya@google.com \and Erik Vee \\ Google Research \\ erikvee@google.com \and Joshua R. Wang \\ Google Research \\ joshuawang@google.com }

\date

\begin{document}

\maketitle

\begin{abstract}
  Caching is a crucial component of many computer systems, so naturally it is a well-studied topic in algorithm design. Much of traditional caching research studies cache management for a single-user or single-processor environment. In this paper, we propose two related generalizations of the classical caching problem that capture issues that arise in a multi-user or multi-processor environment.
  In the \emph{caching with reserves} problem, a caching algorithm is required to maintain at least $k_i$ pages belonging to user $i$ in the cache at any time, for some given reserve capacities $k_i$. In the \emph{public-private caching} problem, the cache of total size $k$ is partitioned into subcaches, a private cache of size $k_i$ for each user $i$ and a shared public cache usable by any user. In both of these models, as in the classical caching framework, the objective of the algorithm is to dynamically maintain the cache so as to minimize the total number of cache misses. 
    
  We show that \emph{caching with reserves} and \emph{public-private caching} models are equivalent up to constant factors, and thus focus on the former. Unlike classical caching, both of these models turn out to be NP-hard even in the offline setting, where the page sequence is known in advance. For the offline setting, we design a 2-approximation algorithm, whose analysis carefully keeps track of a potential function to bound the cost. In the online setting, we first design an $O(\ln k)$-competitive fractional algorithm using the primal-dual framework, and then show how to convert it online to a randomized integral algorithm with the same guarantee.
\end{abstract}

\pagenumbering{arabic}

\section{Introduction}
\label{sec:intro}
Caching is one of the most well-studied problems in online computation and also one of the most crucial components of many computer systems. In the classical caching (also referred to as paging) problem, page requests arrive online and an algorithm must maintain a small set of pages to hold in a cache so as to minimize the number of requests that are not served from the cache. Caching algorithms have been widely studied through the lens of competitive analysis and tight results are known~\cite{achlioptas2000competitive,fiat1991competitive,mcgeoch1991strongly}. Tight algorithms are also known for many generalizations such as weighted paging~\cite{bansal2010towards,bansal2012primal}, generalized caching~\cite{adamaszek2012log,BansalBN12} and paging with rejection penalties~\cite{epstein2015online}. Due to its practical importance, a large number of heuristic algorithms have been proposed such as Least Recently Used (LRU), Least Frequently Used (LFU), CAR~\cite{bansal2004car}, ARC~\cite{megiddo2003arc}, and many others.  Although they do not provide the best worst-case performance, they attempt to maximize the hit rate of the cache on practical instances. However, such traditional caching policies (both theoretical and practical) attempt to optimize the global efficiency of the system and are not necessarily suitable for cache management in a multi-user or multi-processor environment. In many of today’s cloud computing services, caches are shared among all the users utilizing the service and optimizing only for global efficiency can lead to highly undesirable allocation for some users. 
For example, a user who only accesses pages at long intervals may reap no benefit from the cache at all. In this paper, we propose two generalizations of the classical caching problem that are suited for caching in a shared multi-processor environment.

In a multi-user setting, a naive way to guarantee that all users benefit from the cache is to partition the cache among them and effectively maintain separate caches for each user. However, such a system can be extremely inefficient and lead to low overall throughput as the cache can remain underutilized. Instead, a number of recent systems~\cite{wu2022NyxCache,kunjir2017robus,pu2016fairride,yu2019lacs} aim to maximize the global efficiency of the cache while attempting to provide (approximate) \emph{isolation guarantees} to each user, i.e., the cache hit rate for each user is at least as much as what it would be if the user was allocated its own isolated cache (of proportionally smaller size). 
We model the multi-user scenario as the \emph{caching with reserves} problem wherein a caching algorithm is required to maintain at least $k_i$ pages belonging to user $i$ in the cache at any time for some input reserve capacities $k_i$. As in the classical caching framework, the objective of the algorithm is to dynamically maintain the cache so as to minimize the total number of cache misses. The reserve capacities for users provide an implicit isolation guarantee since $k_i$ cache slots are reserved for pages of user $i$.
We remark that when the reserve capacities are all zero, then the problem reduces to classical unweighted caching.

A similar issue arises in the multi-processor setting where we have different ``levels'' of caches. Lower-level caches tend to be smaller and dedicated to a particular processor, while higher-level caches can be used by multiple processors and are larger in size. Consider a system with $m$ separate processors, each of which has its own independent cache. In addition, there is a separate public cache shared by all the processors. We model such a setting as the \emph{public-private caching} problem where a cache of total size $k$ is partitioned into $(m+1)$ subcaches, one private cache for each user and a shared public cache. In contrast with classical caching, in this case cache slots themselves have identities and a page requested by user $i$ cannot be placed in a cache slot that belongs to the private cache of some other user $j$.

\subsection{Our Contributions and Techniques}
We propose and study the \emph{caching with reserves} and \emph{public-private caching} problems. We show that the two problems are equivalent up to constant factors (Section \ref{apx:equiv}).

\begin{restatable}{proposition}{equivalentprop}\label{prop:equivalent}
  If $\A$ is a $c$-competitive online algorithm for \emph{caching with reserves}, then there exists an online algorithm $\A'$ that is $2c$-competitive for \emph{public-private caching}. Similarly, if $\B$ is a $c$-competitive online algorithm for \emph{public-private caching}, then there exists an online algorithm $\B'$ that is $2c$-competitive for \emph{caching with reserves}.
\end{restatable}

Our next set of results considers the \emph{offline} scenario where the entire request sequence is known in advance. Recall that in the classical setting, there is a simple exact solution (Belady's algorithm, which evicts the page that is requested farthest in the future \cite{belady}). In our more complex setting, we show an NP-hardness result (Appendix \ref{apx:np-hardness}). 

The reduction is from 3-SAT. A naive strategy to reduce 3-SAT to our problem is to try to transform boolean variable assignments (e.g. $x_1 = T, x_2 = F$) into the contents of cache at a particular point in time (e.g., agent 1 has its ``true'' page in cache and agent 2 has its ``false'' page in cache). This runs into a stumbling block: to check that a clause is satisfied, one needs to request the relevant pages. Since we only expect one of them to actually be in cache, this provides the opportunity for a cheating solution to swap the contents of cache. Our construction sidesteps this issue by embracing page swapping and instead demanding that a variable assignment be encoded as a particular sequence of page swaps.

\begin{restatable}{theorem}{lowerboundthm}\label{thm:lower-bound}
  Both the offline caching with reserves problem and the offline public-private caching problem are strongly NP-hard.
\end{restatable}

Due to the equivalence of the two models, we focus on \emph{caching with reserves} problem for the rest of the paper. In the offline setting, we give a $2$-approximation algorithm (Section \ref{sec:offline}). It is an adaptation of Belady's algorithm to the multi-agent setting. The analysis utilizes a potential function that was recently proposed to give an alternative proof of optimality for Belady's algorithm \cite{bansal2022learning}.
It tracks how far in the future the cached pages are for the algorithm vs.\ the optimum.

\begin{restatable}{theorem}{offlinethm}\label{thm:offline}
  There is a $2$-approximation algorithm for offline caching with reserves.
\end{restatable}

In the \emph{online} scenario, where the algorithm knows nothing about page requests until they occur, we give a fractional algorithm (which may keep pages fractionally in cache) using the primal-dual framework (Section \ref{sec:online}). 

\begin{restatable}{theorem}{fractionalthm}\label{thm:fractional}
  There is a $2 \ln (k+1)$-competitive fractional  algorithm for online caching with reserves.
\end{restatable}

We also show that the fractional solution can be rounded online in a way that preserves the competitive ratio up to a constant, obtaining an online randomized (integral) algorithm (Section \ref{sec:rounding}).

\begin{restatable}{theorem}{integralthm}\label{thm:integral}
  There is an $O(\ln k)$-competitive integral algorithm for online caching with reserves.
\end{restatable}

\section{Preliminaries and Notation}
\label{sec:prelims}

Let $\U$ be a universe of $n$ pages and $k$ be the number of distinct pages that can be stored in the cache at any time. In the classical caching problem, a sequence of page requests $\sigma = \langle p_1, p_2, \ldots \rangle$, where each $p_t \in \U$,  arrives online and the algorithm is required to maintain a set of at most $k$ pages to be held in the cache at any time. At time $t$, if the currently requested page $p_t$ is not in the cache, then a \emph{cache miss} occurs and the algorithm incurs unit cost. It must then fetch page $p_t$ into the cache possibly by evicting some other page from the cache. An \emph{online} algorithm makes the eviction choice without knowing the future request sequence, whereas an \emph{offline} algorithm is assumed to know the entire request sequence in advance. 

Motivated by applications in multi-processor caching and shared cache systems, we define two new related problems. Let
$\I = \{1, \ldots, m\}$ be a set of $m$ agents and suppose that the universe $\U$ is a disjoint union of pages belonging to each agent, i.e., $\U = \sqcup_{i \in \I} \U(i)$. Let $n_i = |\U(i)|$ be the number of distinct pages owned by agent $i$. For any page $p \in \U(i)$, let $ag(p) = i$ denote the agent that owns page $p$. In the \emph{public-private caching} model, the cache of total size $k$ is subdivided as follows: each agent $i \in \I$ is allocated $k_i$ cache slots and the remaining $k_0 \triangleq k - \sum_{i \in \I} k_i$ slots are \emph{public}.\footnote{We assume throughout the paper that $\sum_{i\in \I} k_i < k$. If $\sum_{i\in \I} k_i = k$, the problem can be solved as $m$ separate instances of classical caching.} In this model, only pages belonging to agent $i$ can be placed in any of the $k_i$ cache slots allocated to agent $i$, while any page can be held in the public slots. As in the traditional caching problem, the goal of the algorithm is to minimize the total number of evictions. 
In the \emph{caching with reserves} model, the cache is not divided, but instead for each agent $i \in \I$, the algorithm is required to maintain at least $k_i$ pages from $\U(i)$ in the cache at any time. To help meet this constraint, it is allowed to begin with dummy pages in its cache that never occur in the actual sequence.

We analyze the online algorithm in terms of its \emph{competitive ratio}. This is the maximum ratio, over all possible problem instances, of the cost incurred by the algorithm to the cost of the optimal offline solution of this instance.
\section{Equivalence of Public-Private Caching and Caching with Reserves}\label{apx:equiv}

We now prove Proposition \ref{prop:equivalent} (restated below for convenience), showing the two models defined in the introduction are equivalent up to constant factors.

\equivalentprop*

\newcommand{\evictions}{\text{evictions}}
\newcommand{\easytransformation}{\tau_e}
\newcommand{\hardtransformation}{\tau_h}
\newcommand{\cropt}{\mathcal{O}_{cr}}
\newcommand{\ppcopt}{\mathcal{O}_{ppc}}

\begin{proof}
  We first explain how to convert back-and-forth between caching strategies for the two problems. Note that both of the following conversions can be done ``online'', i.e. if we know what to evict right now from the cache for one problem, we can determine what to evict right now from the cache for the other problem. The easy direction is turning a public-private caching strategy into a caching with reserves strategy. We will maintain that the cache states in the two problems are identical after every page request. Suppose a page request $p$ comes in. If $p$ is in cache, then we do not evict in either strategy. If it is not, then the public-private caching strategy evicts some page $q$ to make room for it. Our caching with reserves strategy can do so as well while maintaining the reserve constraint, as the following case work shows:
  \begin{itemize}
    \item If $q$ was in a private cache, then $p$ winds up in the same private cache and hence they had the same agent. Hence this agent still has the same number of pages in cache as before for our caching with reserves algorithm.
    \item If $q$ was in a public cache, then $q$'s agent $i$ has at least $k_i+1$ pages in cache before this step (the $k_i$ pages in its private cache
    and $q$). Evicting $q$ hence does not put agent $i$ below its reserve for our caching with reserves algorithm.
  \end{itemize}
  We are now ready to handle the hard case of turning a caching with reserves strategy into a public-private caching strategy. To keep the analysis clean, we cheat slightly and permit the public-private caching strategy to perform extra evictions at any step (but it is still charged for each one). Suppose a page request $p$ comes in. If $p$ is in cache, then we do not evict in either strategy. If it is not, then the caching with reserves strategy evicts some page $q$ to make room for it, which belongs to some agent $i$. We can handle this with at most two evictions, as the following case work shows:
  \begin{itemize}
    \item If $q$ is currently in the public cache, then we evict it and replace it with $p$, making the two caches match again.
    \item If $q$ is currently in a private cache and the agent of $p$ is also $i$, then we again can evict it and replace it with $q$, making the two caches match again.
    \item If $q$ is currently in a private cache and the agent of $p$ is not $i$, then we can infer that agent $i$ has some other page, $q'$, in public cache. This is because $\A$ was able to evict $q$ for $p$ while satisfying agent $i$'s reserve afterwards, so there must have been at least $k_i+1$ pages of agent $i$ in cache at the start of this step. We evict both $q$ and $q'$ and then place $q'$ into agent $i$'s private cache and $p$ into public cache.
  \end{itemize}
  
  We now have conversions between the two problems that approximately preserve the number of evictions, and are ready to prove the main claim. We will use $\easytransformation$ to denote the first transformation, from public-private caching strategies into caching with reserves strategies. We will use $\hardtransformation$ to denote the second transformation, from caching with reserves strategies to public-private caching strategies.
  
  Suppose we have some algorithm $\A$ for caching with reserves, and let $\A' \triangleq \hardtransformation(\A)$. Furthermore, let the optimal solutions to caching with reserves and public-private caching be $\cropt$ and $\ppcopt$, respectively.
  \begin{align*}
    \evictions(\A') &\le 2 \cdot \evictions(\A)                             &\qquad \text{Transformation Guarantee} \\
                    &\le 2c \cdot \evictions(\cropt)                        &\qquad \A \text{ is a $c$-approximation} \\
                    &\le 2c \cdot \evictions(\easytransformation(\ppcopt))  &\qquad \cropt \text{ Optimality} \\
                    &\le 2c \cdot \evictions(\ppcopt)                       &\qquad \text{Transformation Guarantee}
  \end{align*}
  Similarly, suppose we have some algorithm $\B$ for public-private caching and let $\B' \triangleq \easytransformation(\B)$. Again, let the optimal solutions to caching with reserves and public-private caching be $\cropt$ and $\ppcopt$, respectively.
  \begin{align*}
    \evictions(\B') &\le \evictions(\B)                                   &\qquad \text{Transformation Guarantee} \\
                    &\le c \cdot \evictions(\ppcopt)                      &\qquad \B \text{ is a $c$-approximation} \\
                    &\le c \cdot \evictions(\hardtransformation(\cropt))  &\qquad \ppcopt \text{ Optimality} \\
                    &\le 2c \cdot \evictions(\cropt)                      &\qquad \text{Transformation Guarantee}
  \end{align*}
  This completes the proof.
\end{proof}
\section{Offline Caching with Reserves}
\label{sec:offline}

In this section, we present a $2$-approximation algorithm for the offline caching with reserves problem. The algorithm itself can be thought of as a generalization to Belady's classic Farthest-in-Future algorithm \cite{belady}.  Indeed, the algorithm we present reduces to it in the trivial case that $k_i = 0$ for all $i$. 
However, in general, in our setting, there are cases where the farthest-in-future page cannot be evicted due to the reserve constraints.

Our algorithm maintains a partition of the pages in cache into sets $N_i$. For $i>0$, the set $N_i$ consists only of pages for agent $i$; further, we maintain $|N_i| = k_i$ at the beginning of each time step. The set $N_0$ contains the remaining cached pages. 
When a page $p$ associated with agent $i$ arrives and is not already in cache, we insert it into $N_i$. This causes $|N_i| = k_i+1$, so we move the farthest-in-future page from $N_i$ to $N_0$. This, in turn, causes $N_0$ to be too large. So we evict the farthest-in-future page from $N_0$. Notice that we are always allowed to evict such a page, since we maintain $k_i$ pages of agent $i$ in each $N_i$. In the case that $p$ arrives but is already in $N_0$, we first move it to $N_i$, then proceed similarly. In this way, an arriving page always ``passes through'' $N_i$. The full details are in Algorithm \ref{alg:offline}.

Our analysis proving the 2-approximation generalizes a potential argument for Belady's algorithm (proposed recently \cite{bansal2022learning}), but is technically more complicated due to the multi-tiered approach we take. The proof compares our sets $N_i$ with sets $N^*_i$ for the optimal algorithm. (To be more precise, the optimal algorithm maintains a certain set of pages in cache at each time step. We define a partition of these pages into the $N^*_i$ such that each $N^*_i$ consists only of pages from agent $i$, and $|N^*_i| = k_i$ at the beginning of each time step.) We call any page's next request time its \emph{rank}. We define, for any rank $s$, the value $n_i(s)$ to be the number of pages in the set $N_i$ with rank at least $s$ at a given time. Similarly, $n^*_i(s)$ is the number of pages in the set $N^*_i$ with rank at least $s$.\footnote{The sets $N_i$ and $N^*_i$ and the quantities $n_i(s)$ and $n^*_i(s)$ vary over time, but we suppress the dependence on $t$ in the notation for brevity.}

We define our potential function as
%
\[\Phi = \sum_{i=0}^m \phi_i  \text{ ,~~ where } \phi_i = \max_{s}\  [ n_i(s) - n^*_i(s)].\]
Notice that $\phi_i \ge 0$ for every $i$, because when $s$ is larger than the rank of any page in cache, we have $n_i(s) = n^*_i(s) = 0$.
Hence $\Phi \ge 0$.

\begin{algorithm}[t]
\DontPrintSemicolon
Let $N \leftarrow$ set of pages in the cache initially\;
Partition $N = \sqcup_{i=0}^m N_i$ where each $N_i$ (for $i \neq 0)$ contains some arbitrary $k_i$ pages belonging to agent $i$ and $N_0$ contains all the remaining pages\;
Set $rank(q)$, for each page $q$, to the time of $q$'s first request\;
\For{each requested page $p$} {
Let $i=ag(p)$\;
\uIf(\tcc*[f]{Cache hit in a set reserved for $i$.}){$p \in N_i$}{
 Serve page $p$ from cache\;
}
\uElseIf(\tcc*[f]{Cache hit in a set not reserved for $i$.}){$p\in N_0$}{
    Serve page $p$ from cache\;
    \texttt{/* Move $p$ from $N_0$ to $N_i$. */}\;
    $N_i \leftarrow N_i \cup \{p\}$ and $N_0 \leftarrow N_0 \setminus \{p\}$\;
    \texttt{/* Move highest-ranked page from $N_i$ to $N_0$. */}\;
    Let $q_i \in N_i$ be the page in $N_i$ with maximum rank (if $k_i=0$, this will be $p$)\;
    $N_i \leftarrow N_i \setminus \{q_i\}$ and $N_0 \leftarrow N_0 \cup \{q_i\}$
}
\Else(\tcc*[f]{Cache miss.})  {
    \texttt{/* Add $p$ to $N_i$, then move highest-ranked page from $N_i$ to $N_0$. */}\;
    $N_i \leftarrow N_i \cup \{p\}$\;
    Let $q_i \in N_i$ be the page in $N_i$ with maximum rank (if $k_i=0$, this will be $p$)\;
    $N_i \leftarrow N_i \setminus \{q_i\}$ and $N_0 \leftarrow N_0 \cup \{q_i\}$\;
    
    \texttt{/* Evict highest-ranked page from $N_0$. */}\;
    Let $q$ be the page in $N_0$ with maximum rank ($q\neq p$ even if $q_i=p$)\;
     $N_0 \leftarrow N_0 \setminus \{q\}$\;
     \emph{Evict} page $q$, \emph{fetch} page $p$ into cache and serve it\;
 }
 Set $rank(p)$ to the time of $p$'s next request (if none, set it later than the last request)\;
}
\caption{Offline algorithm for caching with reserves.}
\label{alg:offline}
\end{algorithm}

We show that Algorithm \ref{alg:offline} satisfies the requirements of Theorem \ref{thm:offline} (restated below).

\offlinethm*
The proof requires repeated reasoning about how the potential $\Phi$ changes with each step. For example, adding a page to $N_i$ will increase $\phi_i$ by at most 1 (and possibly leave it unchanged). However, adding a page $p$ to $N_i$ whose rank is higher than anything in $N^*_i$ guarantees that $\phi_i$ will increase by exactly 1 (since $n_i(s)$ increases by 1 for every $s \le rank(p)$).

Initially let $N^*_i = N_i$ for all $i$ from 0 to $m$ (the sets $N_i$ are initialized by Algorithm~\ref{alg:offline}). 
Let $ALG$ be the cost incurred by Algorithm \ref{alg:offline} and $OPT$ be the cost incurred by an optimal algorithm. 
Let $\Delta(ALG)$, $\Delta(\Phi)$, $\Delta(OPT)$ be incremental changes in $ALG$, $\Phi$, $OPT$, respectively, with older value subtracted from the newer value.

\begin{restatable}{lemma}{deltaslemma}\label{lem:deltas}
The runs of Algorithm \ref{alg:offline} and of the optimal algorithm on a given sequence of page requests can be partitioned into steps such that for each step, $\Delta(ALG) + \Delta(\Phi) \leq 2\cdot\Delta(OPT)$.  
\end{restatable}

Knowing Lemma \ref{lem:deltas}, 
the approximation factor of 2 now follows from summing over all the incremental steps indexed by $t$, where $\cdot(t)$ is the value of each function after step $t$.
We have $ALG(0) = \Phi(0) = OPT(0) = 0$ initially. By Lemma \ref{lem:deltas}, for each $t$,
\begin{align*}
    ALG(t) - ALG(t-1) + \Phi(t) - \Phi(t-1) &~\leq~ 2\cdot (OPT(t) - OPT(t-1)).\\
    \intertext{Summing over all $t$ (up to the last step $T$) and telescoping,}
    ALG(T) - ALG(0) + \Phi(T) - \Phi(0) &~\leq~ 2\cdot (OPT(T) - OPT(0))\\
    ALG(T) &~\leq~ 2 \cdot OPT(T),
\end{align*}
where the last inequality uses $\Phi(T) \geq 0$.

\begin{proof}[Proof of Lemma \ref{lem:deltas}]
To prove Lemma~\ref{lem:deltas}, 
we break the runs of Algorithm \ref{alg:offline} and the optimal algorithm (together with updates to sets $N^*_i$) into steps, and for each step show that 
$\Delta(ALG) + \Delta(\Phi) \leq 2\cdot\Delta(OPT)$. All the steps below constitute the processing of one request for a page $p$ belonging to agent $i$.
Let $\delta_i(s) = n_i(s) - n_i^*(s)$, so that $\phi_i = \max_s \delta_i(s)$.

\paragraph*{Step 1 (Add $p$ to both $N_i$ and $N^*_i$):}
Update $N_i \leftarrow N_i \cup \{p\}$ and $N^*_i \leftarrow N^*_i \cup \{p\}$. 

Neither $ALG$ nor $OPT$ changes in this step, since we don't evict anything. In addition,
the potential $\Phi$ doesn't increase. To see this, we'll use the fact that the rank of $p$ is the smallest among any page in cache (for our algorithm as well as for the optimal algorithm), since it is the page that has just arrived.
We consider four cases based on whether $N_i$ and $N^*_i$ contained $p$ before this step. 
\begin{itemize}[noitemsep,topsep=0pt]
    \item If both $N_i$ and $N^*_i$ contained $p$  already, then nothing changes.
    \item If neither contained it, then both $n_i(s)$ and $n^*_i(s)$ increase by 1 for all $s\leq rank(p)$, so their difference is unchanged.
    \item If $p$ was newly added only to $N^*_i$, then $\Phi$ can only decrease. 
    \item The remaining case is that $p$ was newly added  only to $N_i$.
Note that since $p$ is the page that was just requested (and its rank hasn't been updated to the next occurrence yet), it has the minimum rank of all pages. We prove that $\Phi$ doesn't increase by showing that before this step, $\phi_i \geq 1$, and after this step, any $\delta_i(\cdot)$ that might have changed are at most 1.
Specifically, before this step, $|N_i|=|N^*_i| = k_i$. Since $N_i$ did not contain $p$, and all other pages have higher rank, before this step we had $n_i(rank(p)+1) = k_i$. Since $ N^*_i$ contained $p$, we had $n^*_i(rank(p)+1) = k_i-1$. Thus, before this step, $\phi_i \geq \delta_i(rank(p)+1) = 1$. After this step, $n_i(s)=k_i+1$, $n^*_i(s)=k_i$, and $\delta_i(s)=1$ for $s\leq rank(p)$
(and $\delta_i(s)$ is unchanged for $s>rank(p)$). Thus, $\Phi$ doesn't increase.
\end{itemize}

\paragraph*{Step 2 (Remove $p$ from both $N_0$ and $N^*_0$):}
Update $N_0 \leftarrow N_0 \setminus \{p\}$ and  $N^*_0 \leftarrow N^*_0 \setminus \{p\}$. 

Again, $ALG$ and $OPT$ don't change since we make no evictions. Further, removing $p$ -- the lowest-ranked page in cache for both our algorithm and the optimal algorithm -- does not increase $\Phi$; the reasoning is similar to above.
\begin{itemize}[noitemsep,topsep=0pt]
 \item If neither $N_0$ nor $N^*_0$ changes, then $\Phi$ remains the same. 
 \item If $p$ is newly removed from both, then $n_0(s)$ and $n^*_0(s)$ decrease by 1 for all $s\leq rank(p)$, and $\delta_0(s)$ for all $s$ are unchanged.
 \item If $p$ is newly removed only from $N_0$, $\Phi$ can only decrease. 
 \item The remaining case is that $p$ was newly removed only from $N^*_0$. Before this step, $|N_0|=|N^*_0| = k_0$. Since $p$ is the page with minimum rank, before the step, $n_0(s)=n^*_0(s)=k_0$ for $s\leq rank(p)$. Also, since before the step $p\notin N_0$ and $p\in N^*_0$, we had $n_0(rank(p)+1)=k_0$ and $n^*_0(rank(p)+1)=k_0-1$, implying $\Phi \geq \delta_0(rank(p)+1)= 1$. After the removal of $p$, $n_0(s)=k_0$, $n^*_0(s)=k_0-1$ and $\delta_0(s)=1$ for $s\leq rank(p)$. Thus, $\Phi$ doesn't increase.
\end{itemize}

\paragraph*{Step 3 (Ensure $|N_i| = |N^*_i| = k_i)$:}
In Step 1, we added $p$ to $N_i$ (resp., $N^*_i$). If it wasn't already there, we increased the size by 1. If that happened, then in this step, we move a page from $N_i$ to $N_0$ to ensure $|N_i| = k_i$ (resp., move from $N^*_i$ to $N^*_0$ to ensure $|N^*_i| = k_i$). Let $q_i$ be the  page in $N_i$ with maximum rank. If $|N_i| = k_i+1$, then $q_i$ is moved to $N_0$, consistent with Algorithm \ref{alg:offline}. We choose which page to move from $N^*_i$ to $N^*_0$ based on the cases below. It could be the page $p$ itself if it is the only one available, the page $q\in N^*_i$ with minimum rank other than $p$ (so it actually has the second-minimum rank in $N^*_i$), or the page $q_i^*\in N^*_i$ with maximum rank. $ALG$ and $OPT$ don't change in this step, and in each case we show that $\Phi$ doesn't increase.
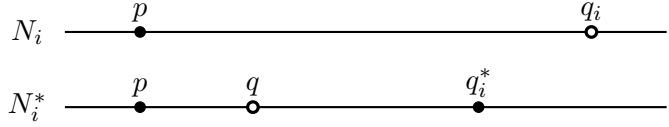
\begin{figure}
\centering
\begin{tikzpicture}

\node at (-0.5, 0) {$N_i^*$};
\draw[black, thick] (0,0) -- (8,0);
\filldraw[black] (1,0) circle (2pt) node[anchor=south]{$p$};
\filldraw[color=black, fill=white, very thick] (2.5,0) circle (2pt) node[anchor=south]{$q$};
\filldraw[black] (5.5,0) circle (2pt) node[anchor=south]{$q_i^*$};

\node at (-0.5, 1) {$N_i$};
\draw[black, thick] (0,1) -- (8,1);
\filldraw[black] (1,1) circle (2pt) node[anchor=south]{$p$};
\filldraw[color=black, fill=white, very thick] (7,1) circle (2pt) node[anchor=south]{$q_i$};

\end{tikzpicture}
\caption{Illustration for the last case of Step 3 in the proof of Lemma \ref{lem:deltas}.}
\label{fig:offline}
\end{figure}
\begin{itemize}[noitemsep,topsep=0pt]
    \item If $k_i=0$, then $N_i = N^*_i = \{p\}$.
    Move $p$ from $N_i$ to $N_0$ and from $N^*_i$ to $N^*_0$.
    
    $\Phi$ is unaffected in this case because for any $s$, $n_i(s)$ changes by the same amount as $n^*_i(s)$, and $n_0(s)$ changes by the same amount as $n^*_0(s)$.
    
    All the cases below assume that $k_i>0$.

    \item If $|N_i| = k_i + 1$ but $|N^*_i| = k_i$, move $q_i$ from $N_i$ to $N_0$. 
    
    We show that when $q_i$ is removed from $N_i$, $\phi_i$ decreases by 1. Since $N_i$ had more pages than $N^*_i$, before this step $\phi_i \geq 1$. Also before this step, $\delta_i(s)\leq 0$ for $s>rank(q_i)$ (since $n_i(s)=0$ for those $s$), so the maximum was not achieved for those values of $s$. And for $s\leq rank(q_i)$, $\delta_i(s)$ decreases by 1 after this step, leading to the decrease of $\phi_i$. Now, when $q_i$ is added to $N_0$, $\phi_0$ increases by at most 1. But this is compensated by the decrease in $\phi_i$, showing that overall $\Phi$ doesn't increase.

    \item If $|N_i| = k_i$ but $|N^*_i| = k_i+1$, move the second-lowest-ranked page $q \in N^*_i$ to $N^*_0$. Note that by our assumption that $k_i>0$, $N^*_i$ has at least two pages.
    
    Adding a page to $N^*_0$ can only decrease the potential. Now we consider the effect on $\phi_i$ of removing $q$ from $N^*_i$. We show that for any $s$ for which $\delta_i(s)$ could have changed, it was negative before this step.
    For any $s>rank(q)$, $\delta_i(s)$ doesn't change. 
    Note that page $p$ has minimum rank in both $N_i$ and $N^*_i$. So, before this step, for $s\leq rank(p)$, $n^*_i(s) = |N^*_i| = k_i+1$ and $n_i(s) = |N_i| = k_i$, so $\delta_i(s)<0$.
    For $s\in (rank(p), rank(q)]$, $n^*_i(s)=k_i$ and $n_i(s)\leq k_i-1$, so again $\delta_i(s)<0$. Thus when $\delta_i(s)$ for $s\leq rank(q)$ increases by 1, it remains at most 0, and does not increase $\Phi$ (which is always at least 0).
    
    \item Recall that $q_i \in N_i$ and $q^*_i \in N^*_i$ are the pages with maximum ranks in the respective sets.
    If $|N_i| = |N^*_i| = k_i+1$ and $rank(q_i) \leq rank(q^*_i)$, move $q_i$ from $N_i$ to $N_0$ and $q^*_i$ from $N^*_i$  to $N^*_0$.
    
    We first consider the removal of $q_i$ from $N_i$ and of $q^*_i$ from $N^*_i$.
    For $s\leq rank(q_i)$, both $n_i(s)$ and $n^*_i(s)$ decrease by 1, so $\delta_i(s)$ doesn't change. 
    For $s>rank(q^*_i)$, $n_i(s)$, $n^*_i(s)$, and $\delta_i(s)$ are unchanged.
    For $s\in (rank(q_i), rank(q^*_i)]$, before this step we had $n_i(s)=0$ and $n_i^*(s)\geq 1$, with $\delta_i(s)\leq -1$. So increasing $\delta_i(s)$ by 1 for these $s$ does not change $\Phi$.
    Now we consider the addition of $q_i$ to $N_0$ and of $q^*_i$ to $N^*_0$. For any $s$, $n^*_0(s)$ increases at least as much as $n_0(s)$ does, so $\Phi$ does not increase.
    
    \item If $|N_i| = |N^*_i| = k_i+1$ and $rank(q^*_i) < rank(q_i)$, move $q_i$ to $N_0$ and the second-lowest-ranked page in $N^*_i$ (call it $q$) to $N^*_0$. Note again that $N^*_i$ has at least two pages. 

    In this case $\phi_0$ may increase by 1, but we show that this is offset by a decrease in $\phi_i$. We analyze what happens for values of $s$ in the intervals separated by three values: $rank(p)<rank(q)<rank(q_i)$ (see Figure \ref{fig:offline}). 
    Before this step, $\delta_i(rank(q_i)) = n_i(rank(q_i)) - n^*_i(rank(q_i)) = 1-0=1$, so $\phi_i \geq 1$. Page $p$ is the page with minimum rank in both $N_i$ and $N^*_i$. For $s\leq rank(p)$, before the step $\delta_i(s)=0$, and it stays 0 after the step. For $s\in (rank(p), rank(q)]$, before the step $n_i^*(s)=|N^*_i|-1=k_i$ and $n_i(s) \leq |N_i|-1=k_i$, so $\delta_i(s)\leq 0$, and it stays that way. For $s>rank(q_i)$, also $\delta_i(s) = 0$ and stays 0.
    Thus, the maximum $\delta_i(s)$ was achieved for some $s\in (rank(q), rank(q_i))$. But in this interval, $n_i(s)$ decreases by 1, while $n^*_i(s)$ stays the same. Thus, the maximum $\delta_i(s)$ decreases by 1, causing $\phi_i$ to also decrease.
\end{itemize}

\paragraph*{Step 4 ($\OPT$ moves):}  If $p$ was in cache, then the optimal algorithm doesn't do anything. Note that in this case, based on previous rearrangements, $|N^*_0|=k_0$. Neither $OPT$ nor $\Phi$ changes.
If $p$ was not in cache, the optimal algorithm fetches $p$ and evicts some page, say $q \in N^*_j$. Then $\Delta(\OPT) = 1$. Also note that in this case the previous steps added $p$ to $\bigcup_\ell N^*_\ell$, resulting in $|N^*_0|=k_0+1$. 
If $j=0$, delete $q$ from $N^*_0$. This restores $|N^*_0|=k_0$ and increases $\Phi$ by at most 1.
If $j\neq 0$, then there must be some $q'\in N^*_0$ belonging to agent $j$ (otherwise it would mean that agent $j$ had only $k_j$ pages in cache, and the optimal algorithm violated reserve sizes by evicting agent $j$'s page).
Move $q'$ from $N^*_0$ to $N^*_j$ and delete $q$ from $N^*_j$. This increases $\Phi$ by at most 2, satisfying the desired inequality.

\paragraph*{Step 5 ($\ALG$ moves):}
If $p$ was in cache, then do nothing. Otherwise, fetch $p$ and evict the page $q$ with maximum rank in $N_0$, also deleting it from $N_0$. In this case, $\Delta(\ALG) = 1$. We show that this is compensated by $\Delta(\Phi)=-1$. Before this step, we had $|N_0|=k_0+1$ but $|N^*_0|=k_0$, so $\phi_0\geq 1$. For $s>rank(q)$, we had $\delta_0(s)\leq 0$, and this doesn't change. So the maximum must have been achieved for $s\leq rank(q)$, and $\delta_0(s)$ for those $s$  decreases by 1.

\paragraph*{Step 6 (Update the rank of $p$):}
At this point, if $k_i = 0$, then $p\in N_0 \cap N^*_0$; otherwise, $p \in N_i \cap N^*_i$.
In either case, changing $rank(p)$ preserves $\delta_0(s)$ and $\delta_i(s)$ for all $s$, so $\Phi$ is unchanged.
\end{proof}

\noindent
This completes the proof of Lemma \ref{lem:deltas} and the proof of Theorem \ref{thm:offline}.

\section{Online Caching with Reserves}
\label{sec:online}

In this section, we design an $O(\log k)$-competitive fractional online algorithm for caching with reserves. In particular, we prove Theorem \ref{thm:fractional}, which is restated here for convenience.
In Section \ref{sec:rounding}, we show that any fractional algorithm for online caching with reserves can be \emph{rounded} to obtain a randomized integral algorithm by losing only a constant factor in the competitive ratio. We remark that our rounding algorithm does not necessarily run in polynomial time.

\fractionalthm*

We begin with the fractional algorithm, which is based on the primal-dual framework and closely follows the analysis of \cite{bansal2012primal}. 
As page requests arrive, the algorithm maintains a feasible solution to the primal LP, which corresponds to its eviction decisions, and an approximately feasible solution to the dual LP. The costs of these two solutions are within a factor 2 of each other. Using LP duality, this results in a bound on the cost incurred by the algorithm compared to the optimum.

\subsection{Notation} 
Consider some fixed page $p \in \U$, and let $t_{p,1} < t_{p,2} < ...$ be the time steps when page $p$ is requested in the online sequence. For any $a \geq 0$, define $I(p,a) = \{t_{p,a} + 1, \ldots, t_{p,a+1}-1\}$ to be the time interval between the $a$th and $(a+1)$th requests for page $p$ (assume that $t_{p,0} = 0$ for all pages). Let $a(p,t)$ be the number of requests to page $p$ that have been seen until time $t$ (inclusive). Hence, by definition, for any time $t$, and any page $p \in \U \setminus \{p_t\}$, we have $t \in I(p, a(p,t))$.
At any time $t$, an agent $i \in \I$ is said to be \emph{\fettered} if exactly $k_i$ pages of agent $i$ are held in cache. Let $\T$ denote the set of \fettered\ agents.\footnote{The set of \fettered\ agents varies with the time $t$, but we suppress the dependence on $t$ for convenience.}

\subsection{Formulation}
We use the variable $x(p, a) \in \{0,1\}$ to denote whether page $p$ is evicted between its $a$th and $(a+1)$th request, i.e., in the interval $I(p,a)$ (where 1 denotes an eviction). We have the following linear programming relaxation and its dual formulation. 
\medskip

\noindent
\hspace{-1em}
\begin{minipage}[t]{0.5\textwidth}
\footnotesize
\begin{center}
    \textbf{Primal LP}
\end{center}
\begin{align}
    & \min \quad \sum_{p \in \U} \sum_{\substack{a \geq 1 : \\ t_{p,a} \leq T}} x(p, a) & \notag \\
    & \text{subject to:} \notag \\
    & \quad \sum_{p \in \U, p \neq p_t} x(p, a(p,t)) \geq n - k & \forall t \label{eq:primalcachesize} \\
    & \sum_{p \in \U(i), p \neq p_t} x(p, a(p,t)) \leq n_i - k_i & \forall t, \forall i \label{eq:primalreserve} \\
    & \hspace{7em} x(p,a) \leq 1 & \forall p, \forall a \label{eq:primalvarbound} \\
    & \hspace{9em}\  x \geq 0
\end{align}
\end{minipage}
\quad\vline\hspace{-1em}
\begin{minipage}[t]{0.5\textwidth}
\footnotesize
\begin{center}
    \textbf{Dual LP}
\end{center}
\begin{align}
&\max \sum_t (n-k) \alpha(t) - \sum_{t,i} (n_i - k_i) \beta(t,i) \notag\\
&\hspace{9em} - \sum_{p,a} \gamma(p,a) \notag \\
&\text{subject to:} \notag \\
&\sum_{t \in I(p,a)} \bigl(\alpha(t) - \beta(t, ag(p)) \bigr) - \gamma(p,a) \notag\\
&\hspace{12em}\leq 1 \quad \forall p, \forall a \label{eq:dualconstraint} \\
& \hspace{9em}\ \alpha,\beta, \gamma \geq 0 
\end{align}
\end{minipage}

\medskip
\medskip

The primal objective simply measures the total number of evictions. The first constraint enforces that at any time $t$ at least $n-k$ pages apart from $p_t$ are outside the cache, which implies that at most $k$ pages (including $p_t$) are inside the cache. The second constraint enforces that at any time, at most $(n_i - k_i)$ pages of agent $i$ are outside cache (which implies that at least $k_i$ pages are inside the cache). Note that this is true even if $p_t \in \U(i)$, since then we know that $p_t$ must be in cache, so of the remaining $n_i - 1$ pages, at least $k_i - 1$ must be in the cache, so the total amount outside cache must be at most $(n_i - 1) - (k_i - 1) = n_i - k_i$. 

\subsection{Algorithm}

For convenience, we assume without loss of generality that the cache is initialized to an arbitrary feasible configuration, i.e., each agent $i$ has some arbitrary $k_i$ pages in the cache, and the rest of the cache has $k_0$ other arbitrary pages.
At each time step, as a new page request arrives online, a new set of constraints for the primal LP are revealed, along with the corresponding new variables in the dual. All newly introduced variables are initialized to zero. Note that after the arrival of a new page request at time $t$, only the primal constraint~\eqref{eq:primalcachesize} may now be unsatisfied; however, \eqref{eq:primalreserve}~and~\eqref{eq:primalvarbound} remain feasible. So to maintain a feasible primal solution, we modify the primal (and dual) variables until Constraint~\eqref{eq:primalcachesize} is satisfied. The online algorithm is required to maintain that all the primal variables $x(p,a)$ only monotonically increase over time.
We remark that the dual solution that we maintain will always be approximately feasibile. The violation in \eqref{eq:dualconstraint} is at most $O(\log k)$ at all times (Claim~\ref{claim:dualfeasible}).

\begin{algorithm}
\caption{Fractional Online Algorithm for  Caching with Reserves}
\label{alg:onlinepd}
\DontPrintSemicolon
Let $\eta \leftarrow \frac{1}{k}$\;
\ForEach{request for page $p$ at time $t$}{
Initialize $x(p, a(p,t)) \leftarrow 0, \alpha(t) \leftarrow 0$, $\gamma(p, a(p,t)) \leftarrow 0$ and $\forall i \in \I, \beta(t, i) \leftarrow 0$\;
\While{primal constraint \eqref{eq:primalcachesize} is unsatisfied}{
    Increase dual variable $\alpha(t)$ by $d\alpha$\;
    \ForEach{\fettered\ agent $i \in \T$}{
    Increase dual variable $\beta(t, i)$ by $d\alpha$\;
    }
    \ForEach{page $q \in \U$} {
    \uIf{$ag(q) \in \T$}{
        Do nothing
    }
    \uElseIf{$x(q, r(q,t)) = 1$}{
    Increase $\gamma(q, r(q,t))$ by $d\alpha$\;
    }
    \Else{
    Increase $x(q, r(q,t))$ by $dx = (x(q, r(q,t)) + \eta) d\alpha$\;
    }
    }
}
}
\end{algorithm}

\subsection{Analysis}

First, we note that the primal solution that we construct is feasible by design.

\begin{claim} 
\label{claim:primalcost}
At all times $t$, we maintain the inequality: Primal Objective $\leq$ $2 \, \cdot$ Dual Objective.
\end{claim}
\begin{proof} 
At time $t = 0$, both the primal and dual solutions are initialized to have an objective of zero. Since the algorithm increases the primal and dual variables in a continuous fashion, consider any infinitesimal time step and let $\Delta P$ and $\Delta D$ denote the change in the primal and dual objectives in this step respectively. It suffices to  show that $\Delta P \leq 2 \cdot \Delta D$ holds at all times.

Let $\T$ denote the set of agents who are \fettered\  during this step. Also partition the set $\U \setminus \{p\}$ into three parts: $T$ is the set of pages belonging to \fettered\  agents, $E = \{q \in \U \setminus T \mid x(q, r(q,t)) = 1\}$ is the set of pages of non-tight agents that have been fully evicted, and $S$ is the remaining set of pages. So we have $|T| + |S| + |E| = n-1$, and $|T| = \sum_{i \in \T} n_i$. We also define $k' := k - \sum_{i \in \T} k_i$.

The change in the dual objective is given by:
\begin{align*}
    \Delta D &= (n-k) d\alpha - \sum_{i \in \T}(n_i - k_i) d\alpha - |E| d\alpha
    = \Bigl(n - k - |T| + \sum_{i \in \T} k_i - |E|\Bigr) d\alpha\\
    &= \Bigl(|S| - \bigl(k - \sum_{i \in \T} k_i \bigr) + 1\Bigr) d\alpha = (|S| - k' + 1) d\alpha
    \intertext{On the other hand, the change in primal objective is given by:}
    \Delta P &= \sum_{q \in S} \bigl(x(q, r(q,t)) + \eta\bigr) d\alpha\\
    &= \Bigl( \sum_{q \in \U \setminus \{p\}} x(q, r(q,t)) - \sum_{q \in T} x(q, r(q,t)) - \sum_{q \in E} x(q, r(q,t)) + |S| \eta \Bigr) d\alpha
    \intertext{Since the variables are updated only as long as constraint~\eqref{eq:primalcachesize} is not satisfied, we can bound the first term in the above expression by $n-k$. All pages in $T$ belong to \fettered\  agents, so we have $\sum_{q \in T} x(q, r(q,t)) = \sum_{i \in \mathcal{T}} (n_i - k_i)$. Lastly, all pages in $E$ have $x(q, r(q,t)) = 1$. So we get:}
    \Delta P & \leq \Bigl(n - k - \sum_{i \in \T} (n_i - k_i) - |E| + |S| \eta \Bigr) d\alpha
    = \Bigl(|S| - \bigl(k - \sum_{i \in \mathcal{F}}k_i\bigr) + 1 + |S|\eta\Bigr) d\alpha \\
    &\leq \bigl(|S| - k' + 1 + |S|/k' \bigr) d\alpha \tag{since $\eta = 1/k \leq 1/k'$} \\
    & \leq 2 (|S| - k' + 1) d\alpha = 2 \cdot \Delta D
\end{align*}
It remains to justify the final inequality, which is equivalent to showing that $|S| \geq k'$. By definition, we have $|S| = n-1 - |E| - |T|$. Since  \eqref{eq:primalcachesize} is violated and \eqref{eq:primalreserve} is tight for $i \in \T$, the following strict inequality holds: 
\[
\sum_{q \in S} x(q, r(q,t)) + |E| + \sum_{i \in \T} (n_i - k_i) = \sum_{q \in S \cup T \cup E} x(q, r(q,t)) < n - k.
\]
Combining the above, we get $|S| > k' - 1$, which implies that $|S| \geq k'$.
\end{proof}

\begin{claim}
\label{claim:dualfeasible}
The dual solution maintained by the algorithm is $O(\log k)$-approximately feasible.
\end{claim}

\begin{proof}
Consider any page $p$ and interval $I(p,a) = \{t_{p,a} + 1, \ldots, t_{p,a+1}-1\}$. We show that the following inequality holds at all times:
\[
\sum_{t \in I(p,a)}(\alpha(t) - \beta(t, ag(p))) - \gamma(p,a) \leq \ln (k+1),
\]
which implies dual feasibility of the solution $(\alpha,\beta,\gamma)$ scaled down by a factor $\ln(k+1)$.

We analyze the changes that occur in the LHS of the above inequality. We interpret the set $I(p,a)$ in an online fashion:  time $t \in \{t_{p,a}+1,\ldots,t_{p,a+1}-1\}$ is included in $I(p,a)$ at the start of the timestep $t$. Note that $x(p, a) = 0$ and the LHS is $0$ at the start of time $t_{p,a}+1$. Over time, as page-requests $p_t (\neq p)$ arrive during times $t \in \{t_{p,a}+1, \ldots, t_{p,a+1}-1\}$, the LHS increases whenever the $\alpha(t)$ variable increases, but there is no corresponding increase in the $\beta(t,ag(p))$ or $\gamma(p,a)$ variables. We couple such increases to increases in the primal variable $x(p,a)$.
Note that $x(p,a)$ gets capped at $1$, and after that $\gamma(p,a)$ is coupled with $\alpha(t)$.

At any infinitesimal step, if some $\alpha(t)$ increases by $d\alpha$, then we have one of three cases. \emph{Case 1:} Agent $ag(p)$ is \fettered\  and $\beta(t, ag(p))$ increases by $d\alpha$; \emph{Case 2:} $x(p, a) = 1$ and $\gamma(p,a)$ increases by $d\alpha$; \emph{Case 3:} $x(p, a)$ increases by $dx = (x(p,a) + \eta)d\alpha$. In the first two cases, the LHS does not change at all, while in the second case, the LHS changes by $d\alpha$. So overall we have 
\begin{align*}
    d(\mathrm{LHS}) &= \left(\frac{1}{x(p,a) + \eta}\right) dx(p,a)
    \intertext{A straightforward integration gives:}
    LHS &= \int_{0}^{X} \left(\frac{1}{x(p,a) + \eta}\right) dx(p,a) \tag{where $X$ is the final value of $x(p,a)$}\\
    &\leq \int_{0}^1 \left(\frac{1}{x(p,a) + \eta}\right) dx(p,a)\\
    &= \left[\ln(x(p,a) + \eta)\right]_0^1 = \ln\Bigl(\frac{1 + \eta}{\eta}\Bigr) = \ln(k+1)
    \qedhere
\end{align*}
\end{proof}

\begin{proof}[Proof of Theorem \ref{thm:fractional}]
The proof follows directly from the two claims above. Let $(x, \alpha, \beta, \gamma)$ denote the primal and dual variables constructed by Algorithm~\ref{alg:onlinepd}, and $(x^*, \alpha^*, \beta^*, \gamma^*)$ be the corresponding variables in the optimal solutions. Using LP duality for the last step, we have:
\begin{align*}
    \sum_{p \in \U} \sum_{\substack{a \geq 1 : \\ t_{p,a} \leq T}}  x(p, a) & \leq 2 \Bigl(\sum_t (n-k) \alpha(t) - \sum_{t,i} (n_i - k_i) \beta(t,i) - \sum_{p,a} \gamma(p,a)\Bigr) \tag{by Claim~\ref{claim:primalcost}} \\
    &\leq 2 \ln(k+1) \Bigl(\sum_t (n-k) \alpha^*(t) - \sum_{t,i} (n_i - k_i) \beta^*(t,i) - \sum_{p,a} \gamma^*(p,a)\Bigr) \tag{by Claim~\ref{claim:dualfeasible}} \\
    &\leq 2 \ln(k+1) \Bigl(\sum_{p \in \U} \sum_{\substack{a \geq 1 : \\ t_{p,a} \leq T}} x^*(p, a)\Bigr)
    \qedhere
\end{align*}
\end{proof}

\section{Rounding} \label{sec:rounding}

We now describe an $O(1)$-approximate rounding scheme for the fractional algorithm of Section \ref{sec:online}, thus proving Theorem \ref{thm:integral}. 

\integralthm*

\begin{proof}
For any time $t = 1,2,\ldots$, the randomized integral algorithm will maintain a distribution $\mu^t$ of cache states 
such that for any page $p$, the probability that page $p$ is not in the cache (of the randomized algorithm) at time $t$ is exactly $x^t(p,r(p,t))$, where $x^t$ denotes $x$ at time $t$.
By the design of our primal-dual algorithm, the $x$-variables never decrease, so the cost incurred by the fractional algorithm to serve page $p_t$ is given by:
\[
\mathrm{cost}(t) := \sum_{p \in \U, p \neq p_t}  \Bigl(x^{t+1}(p, r(p,t)) - x^t(p, r(p,t)) \Bigr).
\]
We will shortly describe how the integral algorithm moves from the distribution $\mu^t$ to $\mu^{t+1}$ while ensuring that the expected number of fetches and evictions is at most $O(\mathrm{cost}(t))$.
We remark that our rounding algorithm does not necessarily run in polynomial time. This is because the support size of $\mu^t$ can be super-polynomial in $|\U|$ and $k$. This is not an issue for online algorithms, so we simply assume that we are maintaining a probability distribution over $O(\binom{|\U|}{k})$ cache states. 

Fix some time $t$. For each page $p \in \U \setminus \{p_t\}$, define $y(p) := 1 -  x^t(p, r(p,t))$ and $y'(p) := 1 - x^{t+1}(p, r(p,t))$ to be the portion of page $p$ that is in the cache at the start of times $t$ and $t+1$, respectively. Also define $y(p_t) = 1 - x^t(p_t, r(p_t,t))$ and $y'(p_t) := 1$; note that the fractional algorithm pays cost $1 - y(p_t)$ to fully fetch $p_t$ into the cache by the end of timestep $t$.
With the above notation, for any page $p \in \U$, we have $\Pr_{C \sim \mu^t}[p \in C] = y(p)$ and $\Pr_{C \sim \mu^{t+1}}[p \in C] = y'(p)$.

To simplify the description of our rounding scheme, we further assume that the changes that occur in the primal solution between states $x^t$ and $x^{t+1}$ do so through a sequence of smaller changes where the $x$-value changes for exactly two pages.
Let $p,q \in \U$ and $\epsilon \in [0,1]$ be such that $y'(p) = y(p) + \epsilon$, $y'(q) = y(q) - \epsilon$, and $y'(p') = y(p')$ for all $p' \in \U \setminus \{p,q\}$.\footnote{Here, $p$ plays the role of page $p_t$ that is fetched into the cache, and $q$ plays the role of pages in $\U \setminus \{p_t\}$ that are evicted to make space for $p_t$.}
Let $\mu, \mu'$ denote distributions over integral cache states that agree with $y$ and $y'$, respectively. 
The cost incurred by the fractional algorithm to move from $y$ to $y'$ is exactly $\epsilon$ (because it only pays for evictions). We now describe how the integral algorithm moves from $\mu$ to $\mu'$ by incurring a cost of at most $4 \epsilon$. To modify a $\delta$ probability measure of the cache-state from $C$ to $C'$, the integral algorithm pays a cost of $\delta \cdot |C \setminus C'|$.
We divide the modification steps into three phases:
\begin{enumerate}
    \item \textbf{Fixing the marginals:} In this phase, we modify the distribution $\mu$ so that for any page $p' \in \U$, $\Pr_{C \sim \mu}[p' \in C]$ changes from $y(p')$ to $y'(p')$. We accomplish this by: (i) adding $p$ to an $\epsilon$ probability measure of cache states from $\mu$ that do not contain $p$; and (ii) removing $q$ from an $\epsilon$ measure of cache states from $\mu$ that contain $q$. The cost incurred in this step is exactly $\epsilon$. 
    
    By the end of this phase, for any (possibly infeasible) cache state $C$ in $\mu$, we have $|C| \in \{k-1,k,k+1\}$. Furthermore, if such a $C$ violates some reserve constraint, then it must have been obtained by removing page $q$ from some other cache state, and so we have $|C| \in \{k-1,k\}$. Let $\epsilon_1 \in [0,1]$ denote the probability measure of cache states with exactly $k-1$ pages. By the description of the modification step, it is clear that $\epsilon_1 \leq \epsilon$ and \emph{exactly} $\epsilon_1$ measure of cache states have cardinality $k+1$. Let $\epsilon_2 \in [0,1]$ denote the measure of cache states that violate some reserve requirement. It is clear that $\epsilon_2 \leq \epsilon$.
    
    \item \textbf{Fixing the size:} In this phase, we match an $\epsilon_1$ measure of cache-states of size $k-1$ with an $\epsilon_1$ measure of cache-states of size $k+1$. Let $C$ and $C'$ denote page-sets of size $k-1$ and $k+1$, respectively, that are matched with some positive measure $\alpha$. Since $|C'| = k+1$, none of the reserve constraints are violated in $C'$ i.e., for all agents $i \in \I$, we have $|C' \cap \U(i)| \geq k_i$. Pick an arbitrary page $p' \in C' \setminus C$. We remove $p'$ from an $\alpha$ measure of state $C'$, and add it to an $\alpha$ measure of state $C$. The cost incurred in this phase is exactly $\epsilon_1 \leq \epsilon$.
    
    By the end of this phase, all cache-states have cardinality exactly $k$. Let $\epsilon_3 \in [0,1]$ denote the measure of cache states that satisfied all reserve constraints at the end of the first phase, but now violate some reserve constraint. By the above discussion, such cache states arise from the removal of page $p' \in C' \setminus C$ from $C'$ (that had size $k+1$), so $\epsilon_3 \leq \epsilon_1$. Overall, exactly $\epsilon_2 + \epsilon_3$ measure of cache states violate some reserve constraint. In fact, every violated cache state violates a single reserve constraint.
    
    \item \textbf{Fixing the violated reserve constraint:} We now fix all violated reserve constraints by matching an $\epsilon_2 + \epsilon_3$ measure of cache states with exactly an $\epsilon_2 + \epsilon_3$ measure of cache states that have an excess in that reserve constraint. More precisely, if $C$ is a cache state that violates the reserve constraint for agent $i \in \I$, then we match an $\alpha > 0$ measure of $C$ with another cache state $C'$ that satisfies $|C' \cap \U(i)| \geq k_i + 1$. Such a matching exists because the fractional solution $y'$ satisfies all reserve constraints and (by the end of the first phase we ensured that) the distribution $\mu$ satisfies the reserve constraint in expectation: for every cache state $C$ with $|C \cap \U(i)| < k$, there must exist another cache state $C'$ with $|C' \cap \U(i)| > k$. We move an arbitrary page $p' \in \U(i) \cap (C' \setminus C)$ from $C'$ to $C$. The cost incurred in this phase is at most $\epsilon_2 + \epsilon_3 \leq 2 \epsilon$. 
    
    At the end of this step, all cache states have size exactly $k$ and satisfy all reserve constraints. The marginal probabilities in the resulting distribution $\mu'$ matches $y'$.
\end{enumerate}

This completes the description of our rounding scheme.
\end{proof}

\bibliographystyle{plainurl}
\bibliography{references_arxiv}
\appendix

\section{NP-hardness of Offline Problems}\label{apx:np-hardness}

We prove Theorem~\ref{thm:lower-bound}, restated here for convenience:

\lowerboundthm*

\newcommand{\page}[2]{p^{#1}_{#2}}
\newcommand{\numpages}{P}
\newcommand{\concat}{\circ}
\newcommand{\clausevariable}[2]{i\left(#1, #2\right)}
\newcommand{\clausevariableoccurrences}[2]{prev\left(#1, #2\right)}
\newcommand{\pattern}[1]{pat(#1)}

\newcommand{\booleanformula}{\varphi}
\newcommand{\cachingproblem}{\textsc{Caching}}
\newcommand{\cachingstrategy}{\mathcal{S}}

\newcommand{\wipegadget}[1]{WIPE$\left(#1\right)$}
\newcommand{\readgadget}[2]{READ$\left(#1, #2\right)$}
\newcommand{\publicgadget}[2]{PUBLIC$\left(#1, #2\right)$}
\newcommand{\clausegadget}[2]{CLAUSE$\left(#1, #2\right)$}
\newcommand{\variablegadget}[1]{VARIABLE$\left(#1\right)$}
\newcommand{\gadgetsize}[1]{\text{size(#1)}}

\begin{proof}
  We will prove that both problems are strongly NP-hard via a single reduction. We reduce from the following variant of $3$-SAT, which is also NP-complete. Given a $3$-CNF Boolean formula $\booleanformula(x_1, ..., x_n)$ with $n$ variables and $m$ clauses, where $n$ is even, is there a satisfying variable assignment where half the variables are true and the other half are false?\footnote{For our reduction, it suffices to pad the formula with $n$ dummy variables that never appear, but it is possible to guarantee the variables actually appear in clauses, e.g.\ \cite{saha2021np}.}

  We will take an instance of this problem $\booleanformula$ and produce a generic caching problem $\cachingproblem(\booleanformula)$ that can be viewed as both a caching with reserves or a public-private caching problem. Our goal is to show that $\booleanformula$ has a half-true half-false satisfying assignment if and only if there is a caching strategy with at most $C$ cache misses (in either the caching with reserves or the public-private caching regime), where $C$ is an integer that we choose later that depends only on $m$ and $n$.

  To reduce the number of relationships between problems and solutions that we need to prove, we will leverage an insight from the proof of Proposition \ref{prop:equivalent}. Namely, when we consider the two caching problems on the same input, we can always transform a public-private caching strategy with $C$ cache misses into a caching with reserves strategy with $C$ cache misses (the other direction is where we actually lost a factor two). Hence to prove our theorem here we only need to establish two facts:
  (i) if there is a half-true half-false satisfying assignment for $\booleanformula$, then there is a public-private caching strategy with at most $C$ cache misses for $\cachingproblem(\booleanformula)$ and (ii) if there is a caching with reserves strategy with at most $C$ cache misses for $\cachingproblem(\booleanformula)$, then there is a half-true half-false satisfying assignment for $\booleanformula$.
  
  Our instance has $(n+4m+3)$ agents\footnote{The last $4m+3$ agents do not actually need to be distinct for the proof but help simplify the presentation.}. The first $n$ agents have unit reserve sizes: $k_1 = k_2 = \cdots = k_n \triangleq 1$; the last $4m+3$ agents, zero reserve sizes: $k_{n+1} = k_{n+2} = \cdots = k_{n+4m+3} \triangleq 0$. The publicly accessible cache has $\frac12 n + 2$ space, so the total cache size is $k \triangleq \frac32 n + 2$.
  
  Regarding pages, we will use $\page{i}{j}$ to denote page $j$ belonging to agent $i$. 
  
  The high-level plan is as follows. We will reason about maximizing the number of cache hits, which is equivalent to minimizing the number of cache misses. In particular, an algorithm may earn a cache hit by permitting a page to occupy cache for the duration between two consecutive requests to that page, but of course is limited by the amount of cache space available.
  
  Each of our first $n$ agents represents a variable of our Boolean formula. Such an agent $i$ has $3\text{deg}(i) + 4$ pages ($\text{deg}(i)$ is the number of clauses containing variable $i$), each of which occurs exactly twice and hence provides a single caching opportunity. Our desired mapping is that setting the associated variable to true corresponds to capitalizing on the caching opportunities of pages whose numbers are congruent to one mod three; false, congruent to two mod three. These true and false subsequences are interwoven so individually they can be safely cached in private cache/reserve, but together they occupy public cache at critical points in the sequence. Additionally, there are pages congruent to zero mod three; these conflict with the one mod three and two mod three sequences in ways that allow us to verify key facts about variables: that they are set to satisfy clauses and that half are true and half are false.
  
  Our reduction involves several gadgets, each of which is just a sequence of particular pages. For each gadget, we will briefly explain its role in the construction, then provide a formal description accompanied by a diagram of the gadget.
  
  \textbf{Public-Cache-Occupying Gadget}. This gadget forces an efficient caching strategy to dedicate $x$ slots of public cache to pages from an agent between $n+1$ and $n+4m+3$. Formally, there are $4m+3$ occurrences of the public-cache-occupying gadget, which each take in a positive integer parameter $x$. We will refer to the $i$th gadget given a parameter of $x$ as \publicgadget{i}{x}, where $i \in [4m+3], x \in \mathbb{Z}^+$. It uses $x$ pages from agent $n+i$ and round-robins over them $C'$ times where $C'$ is an integer that depends only on $n$ and $m$ that we choose later.
  \begin{align*}
  \forall i \in [4m + 3] \qquad
    \text{\publicgadget{i}{x}} &\triangleq
      \left[
        \page{n+i}{1}
        \page{n+i}{2}
        \cdots
        \page{n+i}{x}
      \right]^{C'} \\
    \gadgetsize{\publicgadget{i}{x}} &= C'x
  \end{align*}
  See Figure~\ref{fig:public-gadget} for a visualization of this gadget.
  
  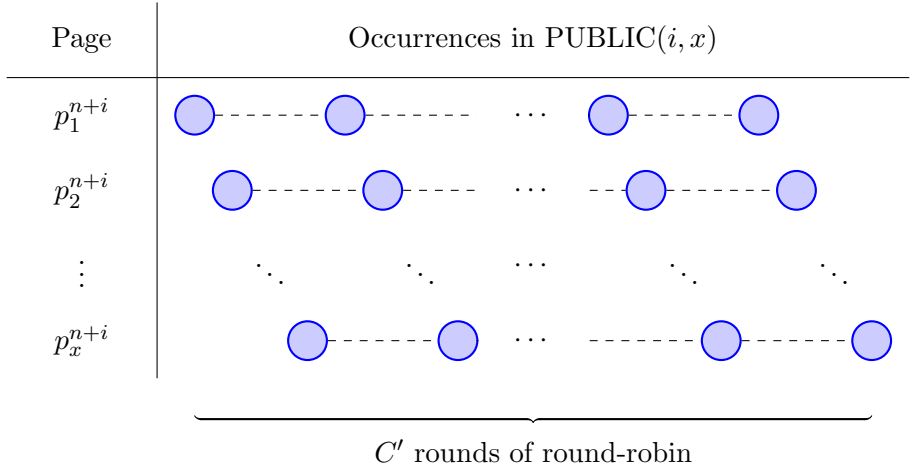
\begin{figure}
\centering
\begin{tikzpicture}[%
  auto,
  page/.style={
    circle,
    draw=blue,
    thick,
    fill=blue!20,
    text width=0.5em,
    align=center,
    minimum height=0.5em
  },
  invis/.style={
    circle,
    draw=none,
    thick,
    fill=none,
    text width=3em,
    align=center,
    minimum height=2em
  }
  ]
  \draw (0, 1.5) -- (0, -3.5);
  \draw (-2, 0.5) -- (10, 0.5);
  
  \node at (-1, 1) {Page};
  \node at (-1, 0) {$\page{n+i}{1}$};
  \node at (-1, -1) {$\page{n+i}{2}$};
  \node at (-1, -2) {$\vdots$};
  \node at (-1, -3) {$\page{n+i}{x}$};
  
  \node (pagesequence) at (5, 1) {Occurrences in \publicgadget{i}{x}};
  
  \node[page] (v11) at (0.5, 0) {};
  \node[page] (v21) at (1, -1) {};
  \node[invis] (v31) at (1.5, -2) {$\ddots$};
  \node[page] (v41) at (2, -3) {};
  
  \node[page] (v12) at (2.5, 0) {};
  \node[page] (v22) at (3, -1) {};
  \node[invis] (v32) at (3.5, -2) {$\ddots$};
  \node[page] (v42) at (4, -3) {};
  
  \node[invis] (v13) at (5, 0) {$\cdots$};
  \node[invis] (v23) at (5, -1) {$\cdots$};
  \node[invis] (v33) at (5, -2) {$\cdots$};
  \node[invis] (v43) at (5, -3) {$\cdots$};
  
  \node[page](v14) at (6, 0) {};
  \node[page] (v24) at (6.5, -1) {};
  \node[invis] (v34) at (7, -2) {$\ddots$};
  \node[page] (v44) at (7.5, -3) {};
  
  \node[page] (v15) at (8, 0) {};
  \node[page] (v25) at (8.5, -1) {};
  \node[invis] (v35) at (9, -2) {$\ddots$};
  \node[page] (v45) at (9.5, -3) {};
  
  \draw[dashed] (v11) -- (v12) -- (v13) -- (v14) -- (v15);
  \draw[dashed] (v21) -- (v22) -- (v23) -- (v24) -- (v25);
  \draw[dashed] (v41) -- (v42) -- (v43) -- (v44) -- (v45);
  
    \draw[decorate, decoration = {calligraphic brace}, thick] (9.5, -4) -- (0.5, -4) node[pos=0.5,below=5pt] {$C'$ rounds of round-robin};
\end{tikzpicture}
\caption{\publicgadget{i}{x} round-robins between $x$ pages with the intention of occupying the $x$ units of the public cache (agent $(n+i)$ has no private cache and is not used elsewhere).}
\label{fig:public-gadget}
\end{figure}
  
  \textbf{Clause Gadget.} The $j$th copy of this gadget checks that clause $j$ is satisfied. To help explain the design of this gadget, consider the clause $(x_5 \lor x_2 \lor \neg x_4)$. We focus on the literal $x_2$. Leading up to this gadget, agent 2 will have one page whose number is congruent to one mod three that has been requested exactly once so far. This page is currently supposed to be in cache if the variable is true, but we actually want a page to be in cache if the variable is false, so that the clause being unsatisfied corresponds to high cache load. Therefore our paging sequence presents the first request for a new page congruent to two mod three for this agent and then the second request for this original page congruent to one mod three. After we arrange the pages of the other two agents to be in a similar state, we insert two requests for a third page congruent to zero mod three, which will take up a unit of public cache if this variable does not satisfy this clause. If none of the variables satisfy the clause, this will use up three units of public cache, which we can detect precisely using a public gadget.
  
  Now that we have given some intuition, we are ready to formally present the clause gadget. There are exactly $m$ occurrences of clause gadgets. We say that a clause $j$'s literal pattern, denoted \pattern{j}, 
  can be one of TTT, TTF, TFF, or FFF depending whether it has zero, one, two, or three negated literals (WLOG we rearrange the literals in each clause so nonnegated literals come first). We will use $\clausevariable{j}{1}$, $\clausevariable{j}{2}$ and $\clausevariable{j}{3}$ to denote the indices of the three variables that appear in the clause. We will also use $\clausevariableoccurrences{j}{\ell}$ to denote how many times $\clausevariable{j}{\ell}$ has appeared in previous clauses:
  \begin{align*}
      \forall \ell \in \{1, 2, 3\} \qquad
      \clausevariableoccurrences{j}{\ell}
        &\triangleq
        \left| \left\{ j' < j \mid x_{\clausevariable{j}{\ell}} \text{ appears in clause } j' \right\} \right|
  \end{align*}

  With these definitions in hand, our clause gadget is defined as follows\footnote{Technically speaking, the way we have defined our clause gadgets means we need each clause to have \emph{exactly} three literals. It is possible to handle smaller clauses as well with this gadget, but this introduces additional indexing complexity.}:
  \begin{align*}
    \text{\clausegadget{j}{TTT}} &\triangleq
      \left[
      \page{\clausevariable{j}{1}}{3\clausevariableoccurrences{j}{1}+2}~
      \page{\clausevariable{j}{2}}{3\clausevariableoccurrences{j}{2}+2}~
      \page{\clausevariable{j}{3}}{3\clausevariableoccurrences{j}{3}+2}
      \concat \text{\publicgadget{4j-2}{\frac12 n + 2}}
      \right.\\
      \\ &\phantom{{}\triangleq{}} \concat
      \left.
      \page{\clausevariable{j}{1}}{3\clausevariableoccurrences{j}{1}+1}~
      \page{\clausevariable{j}{2}}{3\clausevariableoccurrences{j}{2}+1}~
      \page{\clausevariable{j}{3}}{3\clausevariableoccurrences{j}{3}+1}
      \right]
      \\ &\phantom{{}\triangleq{}} \concat
      \left[
      \page{\clausevariable{j}{1}}{3\clausevariableoccurrences{j}{1}+3}~
      \page{\clausevariable{j}{2}}{3\clausevariableoccurrences{j}{2}+3}~
      \page{\clausevariable{j}{3}}{3\clausevariableoccurrences{j}{3}+3}
      \concat \text{\publicgadget{4j-1}{\frac12 n}}
      \right.
      \\ &\phantom{{}\triangleq{}} \concat
      \left.
      \page{\clausevariable{j}{1}}{3\clausevariableoccurrences{j}{1}+3}~
      \page{\clausevariable{j}{2}}{3\clausevariableoccurrences{j}{2}+3}~
      \page{\clausevariable{j}{3}}{3\clausevariableoccurrences{j}{3}+3}
      \right]
      \\ &\phantom{{}\triangleq{}} \concat
      \left[ \text{\publicgadget{4j}{\frac12 n + 2}} \right]
      \\ &\phantom{{}\triangleq{}} \concat
      \left[
      \page{\clausevariable{j}{1}}{3\clausevariableoccurrences{j}{1}+4}~
      \page{\clausevariable{j}{2}}{3\clausevariableoccurrences{j}{2}+4}~
      \page{\clausevariable{j}{3}}{3\clausevariableoccurrences{j}{3}+4}
      \concat \text{\publicgadget{4j+1}{\frac12 n + 2}}
      \right.
      \\ &\phantom{{}\triangleq{}} \concat
      \left.
      \page{\clausevariable{j}{1}}{3\clausevariableoccurrences{j}{1}+2}~
      \page{\clausevariable{j}{2}}{3\clausevariableoccurrences{j}{2}+2}~
      \page{\clausevariable{j}{3}}{3\clausevariableoccurrences{j}{3}+2}
      \right] \\
    \text{\clausegadget{j}{TTF}} &\triangleq
      \left[
      \page{\clausevariable{j}{1}}{3\clausevariableoccurrences{j}{1}+2}~
      \page{\clausevariable{j}{2}}{3\clausevariableoccurrences{j}{2}+2}
      \concat \text{\publicgadget{4j-2}{\frac12 n + 2}} 
      \right. \\
      \\ &\phantom{{}\triangleq{}} \concat
      \left.
      \page{\clausevariable{j}{1}}{3\clausevariableoccurrences{j}{1}+1}~
      \page{\clausevariable{j}{2}}{3\clausevariableoccurrences{j}{2}+1}
      \right]
      \\ &\phantom{{}\triangleq{}} \concat
      \left[
      \page{\clausevariable{j}{1}}{3\clausevariableoccurrences{j}{1}+3}~
      \page{\clausevariable{j}{2}}{3\clausevariableoccurrences{j}{2}+3}~
      \page{\clausevariable{j}{3}}{3\clausevariableoccurrences{j}{3}+3}
      \concat \text{\publicgadget{4j-1}{\frac12 n}}
      \right.
      \\ &\phantom{{}\triangleq{}} \concat
      \left.
      \page{\clausevariable{j}{1}}{3\clausevariableoccurrences{j}{1}+3}~
      \page{\clausevariable{j}{2}}{3\clausevariableoccurrences{j}{2}+3}~
      \page{\clausevariable{j}{3}}{3\clausevariableoccurrences{j}{3}+3}
      \right]
      \\ &\phantom{{}\triangleq{}} \concat
      \left[
      \page{\clausevariable{j}{3}}{3\clausevariableoccurrences{j}{3}+2}
      \concat \text{\publicgadget{4j}{\frac12 n + 2}} \concat
      \page{\clausevariable{j}{3}}{3\clausevariableoccurrences{j}{3}+1}
      \right]
      \\ &\phantom{{}\triangleq{}} \concat
      \left[
      \page{\clausevariable{j}{1}}{3\clausevariableoccurrences{j}{1}+4}~
      \page{\clausevariable{j}{2}}{3\clausevariableoccurrences{j}{2}+4}~
      \page{\clausevariable{j}{3}}{3\clausevariableoccurrences{j}{3}+4}
      \concat \text{\publicgadget{4j+1}{\frac12 n + 2}}
      \right.
      \\ &\phantom{{}\triangleq{}} \concat
      \left.
      \page{\clausevariable{j}{1}}{3\clausevariableoccurrences{j}{1}+2}~
      \page{\clausevariable{j}{2}}{3\clausevariableoccurrences{j}{2}+2}~
      \page{\clausevariable{j}{3}}{3\clausevariableoccurrences{j}{3}+2}
      \right] \\
    \text{\clausegadget{j}{TFF}} &\triangleq
      \left[
      \page{\clausevariable{j}{1}}{3\clausevariableoccurrences{j}{1}+2}
      \concat \text{\publicgadget{4j-2}{\frac12 n + 2}} \concat
      \page{\clausevariable{j}{1}}{3\clausevariableoccurrences{j}{1}+1}
      \right]
      \\ &\phantom{{}\triangleq{}} \concat
      \left[
      \page{\clausevariable{j}{1}}{3\clausevariableoccurrences{j}{1}+3}~
      \page{\clausevariable{j}{2}}{3\clausevariableoccurrences{j}{2}+3}~
      \page{\clausevariable{j}{3}}{3\clausevariableoccurrences{j}{3}+3}
      \concat \text{\publicgadget{4j-1}{\frac12 n}}
      \right.
      \\ &\phantom{{}\triangleq{}} \concat
      \left.
      \page{\clausevariable{j}{1}}{3\clausevariableoccurrences{j}{1}+3}~
      \page{\clausevariable{j}{2}}{3\clausevariableoccurrences{j}{2}+3}~
      \page{\clausevariable{j}{3}}{3\clausevariableoccurrences{j}{3}+3}
      \right]
      \\ &\phantom{{}\triangleq{}} \concat
      \left[
      \page{\clausevariable{j}{2}}{3\clausevariableoccurrences{j}{2}+2}~
      \page{\clausevariable{j}{3}}{3\clausevariableoccurrences{j}{3}+2}
      \concat \text{\publicgadget{4j}{\frac12 n + 2}}
      \right.
      \\ &\phantom{{}\triangleq{}} \concat
      \left[
      \page{\clausevariable{j}{2}}{3\clausevariableoccurrences{j}{2}+1}~
      \page{\clausevariable{j}{3}}{3\clausevariableoccurrences{j}{3}+1}
      \right]
      \\ &\phantom{{}\triangleq{}} \concat
      \left[
      \page{\clausevariable{j}{1}}{3\clausevariableoccurrences{j}{1}+4}~
      \page{\clausevariable{j}{2}}{3\clausevariableoccurrences{j}{2}+4}~
      \page{\clausevariable{j}{3}}{3\clausevariableoccurrences{j}{3}+4}
      \concat \text{\publicgadget{4j+1}{\frac12 n + 2}}
      \right.
      \\ &\phantom{{}\triangleq{}} \concat
      \left.
      \page{\clausevariable{j}{1}}{3\clausevariableoccurrences{j}{1}+2}~
      \page{\clausevariable{j}{2}}{3\clausevariableoccurrences{j}{2}+2}~
      \page{\clausevariable{j}{3}}{3\clausevariableoccurrences{j}{3}+2}
      \right] \\
    \text{\clausegadget{j}{FFF}} &\triangleq
      \left[ \text{\publicgadget{4j-2}{\frac12 n + 2}} \right]
      \\ &\phantom{{}\triangleq{}} \concat
      \left[
      \page{\clausevariable{j}{1}}{3\clausevariableoccurrences{j}{1}+3}~
      \page{\clausevariable{j}{2}}{3\clausevariableoccurrences{j}{2}+3}~
      \page{\clausevariable{j}{3}}{3\clausevariableoccurrences{j}{3}+3}
      \concat \text{\publicgadget{4j-1}{\frac12 n}}
      \right.
      \\ &\phantom{{}\triangleq{}} \concat
      \left.
      \page{\clausevariable{j}{1}}{3\clausevariableoccurrences{j}{1}+3}~
      \page{\clausevariable{j}{2}}{3\clausevariableoccurrences{j}{2}+3}~
      \page{\clausevariable{j}{3}}{3\clausevariableoccurrences{j}{3}+3}
      \right]
      \\ &\phantom{{}\triangleq{}} \concat
      \left[
      \page{\clausevariable{j}{1}}{3\clausevariableoccurrences{j}{1}+2}~
      \page{\clausevariable{j}{2}}{3\clausevariableoccurrences{j}{2}+2}~
      \page{\clausevariable{j}{3}}{3\clausevariableoccurrences{j}{3}+2}
      \concat \text{\publicgadget{4j}{\frac12 n + 2}}
      \right.
      \\ &\phantom{{}\triangleq{}} \concat
      \left[
      \page{\clausevariable{j}{1}}{3\clausevariableoccurrences{j}{1}+1}~
      \page{\clausevariable{j}{2}}{3\clausevariableoccurrences{j}{2}+1}~
      \page{\clausevariable{j}{3}}{3\clausevariableoccurrences{j}{3}+1}
      \right]
      \\ &\phantom{{}\triangleq{}} \concat
      \left[
      \page{\clausevariable{j}{1}}{3\clausevariableoccurrences{j}{1}+4}~
      \page{\clausevariable{j}{2}}{3\clausevariableoccurrences{j}{2}+4}~
      \page{\clausevariable{j}{3}}{3\clausevariableoccurrences{j}{3}+4}
      \concat \text{\publicgadget{4j+1}{\frac12 n + 2}}
      \right.
      \\ &\phantom{{}\triangleq{}} \concat
      \left.
      \page{\clausevariable{j}{1}}{3\clausevariableoccurrences{j}{1}+2}~
      \page{\clausevariable{j}{2}}{3\clausevariableoccurrences{j}{2}+2}~
      \page{\clausevariable{j}{3}}{3\clausevariableoccurrences{j}{3}+2}
      \right]
  \end{align*}
  
  See Figure~\ref{fig:clause-gadget} for a visualization of this gadget. Each such gadget adds the following number of page requests to the sequence:
  \begin{align*}
    \gadgetsize{\clausegadget{j}{\pattern{j}}}
      &= \underbrace{ C'\left(2n + 6\right) }_{\text{\publicgadget{\cdot}{\cdot}}} + 18
       = 2 C'n + 6C' + 18
  \end{align*}
  
  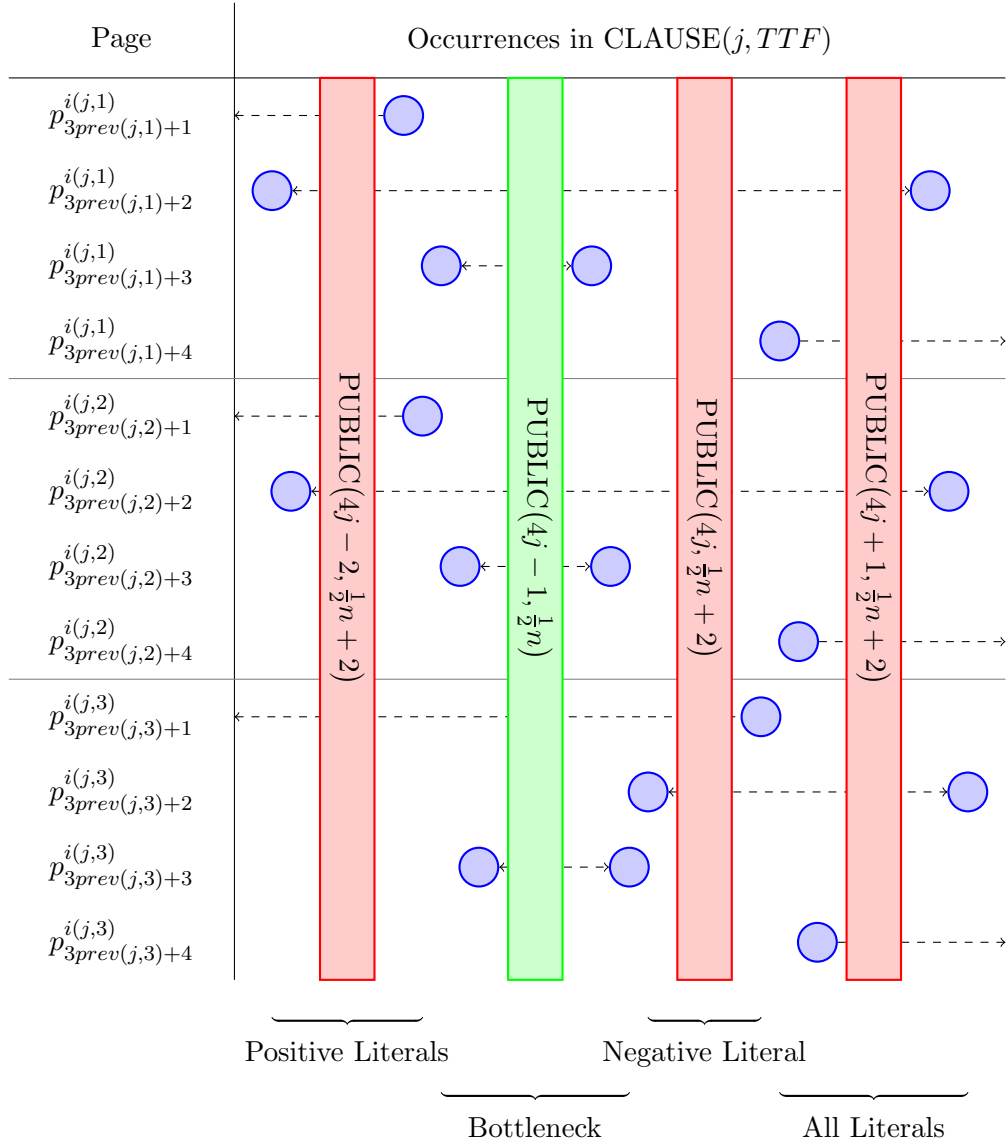
\begin{figure}
\centering
\begin{tikzpicture}[%
  auto,
  page/.style={
    circle,
    draw=blue,
    thick,
    fill=blue!20,
    text width=0.5em,
    align=center,
    minimum height=0.5em
  },
  wipegadget/.style={
    rectangle,
    draw=red,
    thick,
    fill=red!20,
    minimum width=12cm,
    minimum height=0.5cm,
    rotate=-90
  },
  readgadget/.style={
    rectangle,
    draw=green,
    thick,
    fill=green!20,
    minimum width=12cm,
    minimum height=0.5cm,
    rotate=-90
  }
  ]
  \draw (0, 1.5) -- (0, -11.5);
  \draw (-3, 0.5) -- (10.25, 0.5);
  \draw[color=black!50] (-3, -3.5) -- (10.25, -3.5);
  \draw[color=black!50] (-3, -7.5) -- (10.25, -7.5);
  
  \node at (-1.5, 1) {Page};
  \node at (-1.5, 0) {$\page{\clausevariable{j}{1}}{3\clausevariableoccurrences{j}{1}+1}$};
  \node at (-1.5, -1) {$\page{\clausevariable{j}{1}}{3\clausevariableoccurrences{j}{1}+2}$};
  \node at (-1.5, -2) {$\page{\clausevariable{j}{1}}{3\clausevariableoccurrences{j}{1}+3}$};
  \node at (-1.5, -3) {$\page{\clausevariable{j}{1}}{3\clausevariableoccurrences{j}{1}+4}$};
  \node at (-1.5, -4) {$\page{\clausevariable{j}{2}}{3\clausevariableoccurrences{j}{2}+1}$};
  \node at (-1.5, -5) {$\page{\clausevariable{j}{2}}{3\clausevariableoccurrences{j}{2}+2}$};
  \node at (-1.5, -6) {$\page{\clausevariable{j}{2}}{3\clausevariableoccurrences{j}{2}+3}$};
  \node at (-1.5, -7) {$\page{\clausevariable{j}{2}}{3\clausevariableoccurrences{j}{2}+4}$};
  \node at (-1.5, -8) {$\page{\clausevariable{j}{3}}{3\clausevariableoccurrences{j}{3}+1}$};
  \node at (-1.5, -9) {$\page{\clausevariable{j}{3}}{3\clausevariableoccurrences{j}{3}+2}$};
  \node at (-1.5, -10) {$\page{\clausevariable{j}{3}}{3\clausevariableoccurrences{j}{3}+3}$};
  \node at (-1.5, -11) {$\page{\clausevariable{j}{3}}{3\clausevariableoccurrences{j}{3}+4}$};
  
  \node at (5.125, 1) {Occurrences in \clausegadget{j}{TTF}};
  
  \node[page] (v121) at (0.5, -1) {};
  \node[page] (v221) at (0.75, -5) {};
  \node[page] (v112) at (2.25, 0) {};
  \node[page] (v212) at (2.5, -4) {};

  \node[page] (v131) at (2.75, -2) {};
  \node[page] (v231) at (3, -6) {};
  \node[page] (v331) at (3.25, -10) {};
  \node[page] (v132) at (4.75, -2) {};
  \node[page] (v232) at (5, -6) {};
  \node[page] (v332) at (5.25, -10) {};
  
  \node[page] (v321) at (5.5, -9) {};
  \node[page] (v312) at (7, -8) {};
  
  \node[page] (v141) at (7.25, -3) {};
  \node[page] (v241) at (7.5, -7) {};
  \node[page] (v341) at (7.75, -11) {};
  \node[page] (v122) at (9.25, -1) {};
  \node[page] (v222) at (9.5, -5) {};
  \node[page] (v322) at (9.75, -9) {};

  \draw[dashed, <-] (0, 0) -- (v112);
  \draw[dashed, <->] (v121) -- (v122);
  \draw[dashed, <->] (v131) -- (v132);
  \draw[dashed, ->] (v141) -- (10.25, -3);
  \draw[dashed, <-] (0, -4) -- (v212);
  \draw[dashed, <->] (v221) -- (v222);
  \draw[dashed, <->] (v231) -- (v232);
  \draw[dashed, ->] (v241) -- (10.25, -7);
  \draw[dashed, <-] (0, -8) -- (v312);
  \draw[dashed, <->] (v321) -- (v322);
  \draw[dashed, <->] (v331) -- (v332);
  \draw[dashed, ->] (v341) -- (10.25, -11);

  \node[wipegadget] at (1.5, -5.5) {\publicgadget{4j-2}{\frac12 n + 2}};
  \node[readgadget] at (4, -5.5) {\publicgadget{4j-1}{\frac12 n}};
  \node[wipegadget] at (6.25, -5.5) {\publicgadget{4j}{\frac12 n + 2}};
  \node[wipegadget] at (8.5, -5.5) {\publicgadget{4j+1}{\frac12 n + 2}};
  
  \draw[decorate, decoration = {calligraphic brace}, thick] (2.5, -12) -- (0.5, -12) node[pos=0.5,below=5pt] {Positive Literals};
  \draw[decorate, decoration = {calligraphic brace}, thick] (5.25, -13) -- (2.75, -13) node[pos=0.5,below=5pt] {Bottleneck};
  \draw[decorate, decoration = {calligraphic brace}, thick] (7, -12) -- (5.5, -12) node[pos=0.5,below=5pt] {Negative Literal};
  \draw[decorate, decoration = {calligraphic brace}, thick] (9.75, -13) -- (7.25, -13) node[pos=0.5,below=5pt] {All Literals};
\end{tikzpicture}
\caption{\clausegadget{j}{TTF} has a bottleneck region where its three variables might temporarily occupy a unit of public cache. If none of the variables has an appropriate value to satisfy the clause, then \publicgadget{j}{\frac12 n} will incur large number of cache misses.}
\label{fig:clause-gadget}
\end{figure}

  \textbf{Variable Gadgets.} This pair of gadgets, \variablegadget{T} and \variablegadget{F}, help enforce that there are at most $\frac12 n$ variables set to true or false, respectively. \variablegadget{T} is the very first gadget in our overall page request sequence. It requests the first pages for every variable agent. Immediately after this, there is a moment where only the one mod three subsequences have used up private caches. We insert some additional zero mod three pages to make these agents overflow into public cache and also only allow $\frac12 n$ to do so. This bounds the number of one mod three subsequences that can be chosen; i.e.\ only half the variables may be true.
  
  Similarly, \variablegadget{F} is the very last gadget in our overall page request sequence. It requests the last two mod three page for every agent. During this request, we insert some additional zero mod three pages to only allow $\frac12 n$ of these to be chosen; i.e.\ only half the variables may be false.
  
  Formally, our variable gadgets are defined as follows: 
  \begin{align*}
    \text{\variablegadget{T}} &=
      \page{1}{1}~ \page{2}{1} \cdots \page{n}{1}
      \\ &\phantom{{}={}} \concat
      \page{1}{0}~ \page{2}{0} \cdots \page{n}{0}
      \\ &\phantom{{}={}} \concat
      \text{\publicgadget{1}{2}}
      \\ &\phantom{{}={}} \concat
      \page{1}{0}~ \page{2}{0} \cdots \page{n}{0} \\
    \text{\variablegadget{F}} &=
      \page{1}{3\deg(1)+2}~ \page{2}{3\deg(2)+2} \cdots \page{n}{3\deg(n)+2}
      \\ &\phantom{{}={}} \concat
      \text{\publicgadget{4m+2}{\frac12 n + 2}}
      \\ &\phantom{{}={}} \concat
      \page{1}{3\deg(1)+1}~ \page{2}{3\deg(2)+1} \cdots \page{n}{3\deg(n)+1}
      \\ &\phantom{{}={}} \concat
      \page{1}{3\deg(1)+3}~ \page{2}{3\deg(2)+3} \cdots \page{n}{3\deg(n)+3}
      \\ &\phantom{{}={}} \concat
      \text{\publicgadget{4m+3}{2}}
      \\ &\phantom{{}={}} \concat
      \page{1}{3\deg(1)+3}~ \page{2}{3\deg(2)+3} \cdots \page{n}{3\deg(n)+3}
      \\ &\phantom{{}={}} \concat
      \page{1}{3\deg(1)+2}~ \page{2}{3\deg(2)+2} \cdots \page{n}{3\deg(n)+2} \\
    \gadgetsize{\variablegadget{T}} &= 3n + 2C' \\
    \gadgetsize{\variablegadget{F}} &= 5n + 2C' + C'\left( \frac12 n + 2\right)
  \end{align*}
  See Figure~\ref{fig:variable-gadget} for a visualization of this gadget.
  
  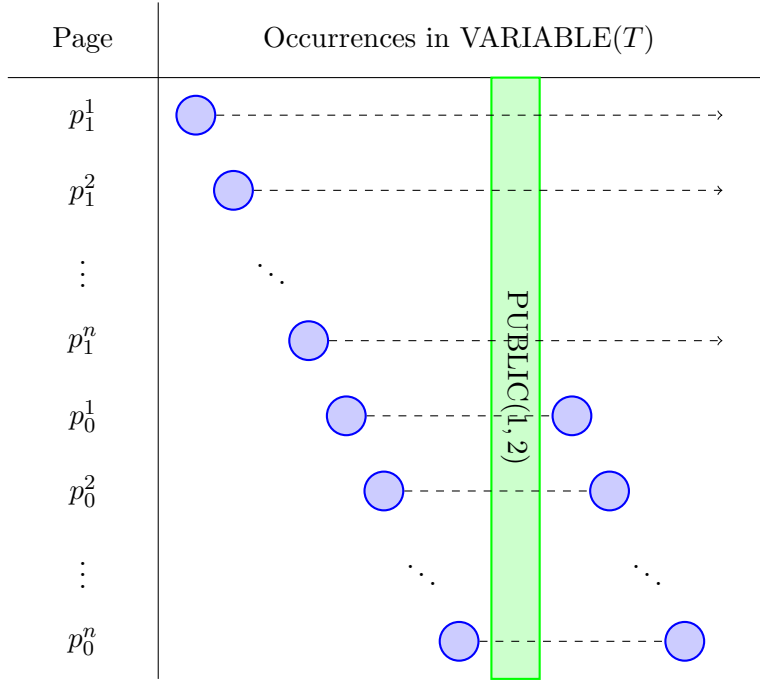
\begin{figure}
\centering
\begin{tikzpicture}[%
  auto,
  page/.style={
    circle,
    draw=blue,
    thick,
    fill=blue!20,
    text width=0.5em,
    align=center,
    minimum height=0.5em
  },
  invis/.style={
    circle,
    draw=none,
    thick,
    fill=none,
    text width=5em,
    align=center,
    minimum height=2em
  },
  publicgadget/.style={
    rectangle,
    draw=green,
    thick,
    fill=green!20,
    minimum width=8cm,
    minimum height=0.5cm,
    rotate=-90
  }
  ]
  \draw (0, 1.5) -- (0, -7.5);
  \draw (-2, 0.5) -- (8, 0.5);
  
  \node at (-1, 1) {Page};
  \node at (-1, 0) {$\page{1}{1}$};
  \node at (-1, -1) {$\page{2}{1}$};
  \node at (-1, -2) {$\vdots$};
  \node at (-1, -3) {$\page{n}{1}$};
  \node at (-1, -4) {$\page{1}{0}$};
  \node at (-1, -5) {$\page{2}{0}$};
  \node at (-1, -6) {$\vdots$};
  \node at (-1, -7) {$\page{n}{0}$};
  
  \node at (4, 1) {Occurrences in \variablegadget{T}};
  \node[page] (v112) at (0.5, 0) {};
  \node[page] (v122) at (1, -1) {};
  \node[invis] (v132) at (1.5, -2) {$\ddots$};
  \node[page] (v142) at (2, -3) {};
  \node[page] (v312) at (2.5, -4) {};
  \node[page] (v322) at (3, -5) {};
  \node[invis] (v332) at (3.5, -6) {$\ddots$};
  \node[page] (v342) at (4, -7) {};
  
  \node[publicgadget] (r1) at (4.75, -3.5) {\publicgadget{1}{2}};
  
  \node[page] (v313) at (5.5, -4) {};
  \node[page] (v323) at (6, -5) {};
  \node[invis] (v333) at (6.5, -6) {$\ddots$};
  \node[page] (v343) at (7, -7) {};
  
  \draw[dashed, ->] (v112) -- (7.5, 0);
  \draw[dashed, ->] (v122) -- (7.5, -1);
  \draw[dashed, ->] (v142) -- (7.5, -3);
  \draw[dashed] (v312) -- (v313);
  \draw[dashed] (v322) -- (v323);
  \draw[dashed] (v342) -- (v343);

\end{tikzpicture}
\caption{\variablegadget{T} serves to ``initialize'' the variables and enforce at most half of them are set to true.}
\label{fig:variable-gadget}
\end{figure}

  \textbf{Putting it All Together.} Our overall page request sequence simply consists of the concatenation of these gadgets, as follows:
  \begin{align*}
    \cachingproblem(\booleanformula)
      &\triangleq \text{\variablegadget{T}}
      \concat \text{\clausegadget{1}{\pattern{1}}} \concat \cdots
      \concat \text{\clausegadget{m}{\pattern{m}}}
      \\ &\vphantom{{}\triangleq{}}
      \concat \text{\variablegadget{F}} \\
    \gadgetsize{$\cachingproblem(\booleanformula)$} &=
      \underbrace{\left[ 3n + 2C' \right]}_{\text{\variablegadget{T}}}
                + \underbrace{\left[ \sum_j 2 C'n + 6C' + 18 \right]}_{\text{\clausegadget{j}{\pattern{j}}} \forall j}
                + \underbrace{\left[ 5n + 2C' + C'\left(\frac12 n + 2\right) \right]}_{\text{\variablegadget{F}}} \\
      &= 2C'mn + 6C'm + \frac12 C'n + 6C' + 18m + 8n
  \end{align*}
  
  \textbf{Correctness.} We have finished presenting the construction and will now reason about its correctness, i.e.\ we want to show that if the original formula $\booleanformula$ is satisfiable by a half-true half-false assignment, then $\cachingproblem(\booleanformula)$ as a public-private caching problem has a strategy that has at most $C$ cache misses, and if $\cachingproblem(\booleanformula)$ as a caching with reserves problem has a strategy that has at most $C$ cache misses then the original formula $\booleanformula$ is satisfiable by a half-true half-false assignment (and we have not chosen $C$ yet). Here is a quick review of the caching opportunities available in our page request sequence:
  \begin{itemize}
    \item For each variable $x_i$, there is a corresponding agent $i$ that has $3\deg(i)+4$ unique numbered pages, all of which allow for a single cache hit.
    \item For each \publicgadget{i}{x} ($4m+3$ in total), there is a corresponding agent $n+i$ with $x$ pages, which each allow for $C'-1$ cache hits.
  \end{itemize}
  
  Now we describe how to convert a (satisfying, half-true, half-false) assignment $\vec{x}$ into a public-private caching strategy $\cachingstrategy(\vec{x})$. The caching strategy will make the following decisions, which are enough to determine the entire strategy:
  \begin{itemize}
    \item For each variable $x_i$, if $x_i$ is set to true then we choose cache hits for the pages congruent to zero or one mod three of the corresponding agent $i$. Pages congruent to one mod three are always in private cache, and pages congruent to zero mod three are in private cache if they do not overlap with pages congruent to one mod three. If $x_i$ is set to false, then we do the same thing with two mod three in place of one mod three.
    \item For each \publicgadget{i}{x}, we choose all cache hits for its corresponding agent $n+i$.
  \end{itemize}
  
  We now show that this caching strategy $\cachingstrategy(\vec{x})$ is valid, i.e. it never exceeds any private cache or public cache. The former fact is easy to see; the set of pages congruent to one do not overlap by construction (in particular, in the design of our clause gadget), and neither do the set of pages congruent to two mod three. Since we only try to fit one of those into private cache and then flexibly fit as many as possible zero mod three pages into private cache, we cannot use more than one unit of private cache per variable agent. The public cache accounting is more complex, and we will reason bottom-up over the gadgets we have presented.
  
  The bottom-most gadget is \publicgadget{i}{x}. We can safely achieve all these cache hits as long as there are $x$ slots of public cache for the duration of the gadget, which we will verify when reasoning about the higher-level gadgets.
  
  Now, consider some \clausegadget{j}{\pattern{j}}. Recall that each variable is set to either true or false, and in the former case we choose its pages congruent to zero/one mod three and in the latter case we choose its pages congruent to zero/two mod three. During \publicgadget{4j-2}{\frac12n + 2}, \publicgadget{4j}{\frac12n + 2}, and \publicgadget{4j+1}{\frac12n + 2}, our caching strategy does not have any variable agent pages in public cache and therefore all of the $\frac12 + n$ public cache slots are available to handle these gadgets. During \publicgadget{4j-1}{\frac12 n}, we know that since we had a satisfying assignment, one of the agents is able to fit its zero mod three page into private cache and hence at most two units of public cache are occupied, leaving $\frac12 n$ public cache for this gadget. For thoroughness, we remember to consider that there needs to be a slot of cache to temporarily hold any page in this gadget, whether we plan to capitalize on its caching opportunity or not. However, this is easily possible because we have accounted for the pages during public-cache occupying subgadgets and outside of that we definitely have (at least two) slots of public cache space.
  
  We finish by considering \variablegadget{T} and \variablegadget{F}. For \variablegadget{T}, we observe that since we are picking the one mod three (true) subsequence for at most $\frac12 n$ variables, we get to put $\frac12 n$ pages of the form $\page{i}{0}$ into private caches and only have $\frac12 n$ such pages occupy public cache, leaving two slots for \publicgadget{1}{2}.
  
  The reasoning is the similar for \variablegadget{F}; we pick the two mod three (false) subsequence for at most $\frac12 n$ variables and hence we can put $\frac12 n$ pages of the form $\page{i}{3\deg(i)+3}$ into private caches. We hence have only $\frac12 n$ such pages occupy public cache, leaving two slots for \publicgadget{4m+3}{2}. The additional \publicgadget{4m+2}{\frac12 n + 2} is safe for the same reasons as the matching subgadgets in the clause gadgets; we picked only one mod three or two mod three pages for each variable and hence do not use any public cache during this subgadget. We continue to be thorough and double-check that there is a slot of cache to temporarily hold each of the pages in this gadget. Again, the pages in public-cache-occupying subgadgets have already been accounted for and outside of that we definitely have (at least two) slots of public cache space.
  
  Now that we have a feasible caching strategy, let us count the number of cache hits it achieves: 
  \begin{align*}
    \text{cache-hits}\left(\cachingstrategy(\vec{x})\right)
      &= \underbrace{\left[ \sum_{i=1}^n 2\deg(i) + 3 \right]}_{\text{Agents } i = 1, 2, ..., n}
       + \underbrace{\left[ (C'-1)\left(2mn + 6m + \frac12 n + 6\right) \right]}_{\text{\publicgadget{\cdot}{\cdot}}} \\
      &= \left[ 2m + 3n \right] + \left[ (C'-1)\left(2 mn + 6m + \frac12 n + 6\right) \right] \\
      &= 2C'mn + 6C'm + \frac12 C'n + 6C' - 2mn - 4m + \frac52 n - 6
  \end{align*}
  
  We are now ready to choose $C$ to be \gadgetsize{$\cachingproblem(\booleanformula)$} minus this quantity. 
  \begin{align*}
    C &\triangleq \underbrace{\left[2C'mn + 6C'm + \frac12 C'n + 6C' + 18m + 8n \right]}_{\gadgetsize{$\cachingproblem(\booleanformula)$}} \\
      &\phantom{{}\triangleq{}}- \left[ 2C'mn + 6C'm + \frac12 C'n + 6C' - 2mn - 4m + \frac52 n - 6 \right] \\
      &= 2mn + 22m + \frac{11}{2} n + 6
  \end{align*}

  We now want to show that if there is a caching-with-reserves strategy $\cachingstrategy$ with this many cache-misses (and hence $\gadgetsize{$\cachingproblem(\booleanformula)$} - C$ cache-hits, we can recover a satisfying, half-true, half-false assignment to $\booleanformula$. We need to reason about how these cache hits are being achieved. We already know that the maximum number of cache hits between the public-cache-occupying subgadgets is
  \[
    \left[ (C'-1)\left(2mn + 6m + \frac12 n + 6\right) \right]
  \]
  because that represents taking every caching opportunity in those subgadgets. Our concern is that perhaps one could obtain extra cache hits on agents $i = 1, 2, ..., n$ by sacrificing some cache hits on these public-cache-occupying subgadgets. However, these subgadgets have been engineered to prevent exactly this; the repeated round-robin means that a caching strategy that is even one slot of public cache space short will incur multiple extra misses. Recall that with $x$ spare public cache space, a caching strategy can handle \publicgadget{i}{x} with only $x$ cache misses (on the first appearance of each page). What happens if we only have $x-1$ spare public cache space instead? The caching strategy must still fault on the initial appearance of each page. In addition, between each consecutive set of requests to all $k$ pages, the algorithm can only keep $k-1$ of them in memory and hence gets a cache miss on at least one page in the latter set of requests. Since there are $C'$ sets of requests, this means we incur at least $C'-1$ additional cache misses\footnote{This is an underestimate, e.g. for $x = 2$ being one slot of public cache short means the caching strategy gets a cache miss on every single page request!}. In other words, the caching strategy does not free up a slot of public cache space during one of these subgadgets unless it incurs at least $C'-1$ additional cache misses. We want to make this not worth it, so we are now ready to choose:
  \begin{align*}
    C' &\triangleq 3m + 4n + 2 \\
    C' - 1 &> 3m + 4n \\
           &= \sum_{i=1}^n 3\deg(i) + 4
  \end{align*}
  In other words, getting an additional slot of public cache space during any subgadget costs more cache misses than all nonsubgadget caching opportunities combined. Hence $\cachingstrategy$ cannot do so and must allocate adequate public cache space to all public-cache-occupying subgadgets.
  
  We are now ready to reason about the number of cache hits among the ``variable'' agents ($i = 1, 2, ..., n$). We know this caching strategy $\cachingstrategy$ achieves at least $3m + 4n$ cache hits among these agents. How are these cache hits distributed between the agents? We claim that agent $i$ does not permit more than $2\deg(i) + 3$ cache hits. For the sake of contradiction, suppose $\cachingstrategy$ achieved more than $2\deg(i) + 3$ cache hits for some agent $i \in [n]$. Subtracting the zero mod three pages, this means there are more than $\deg(i) + 1$ cache hits among the $2\deg(i) + 2$ one and two mod three pages. But that means that if we wrote down all these pages in sorted order, we would have to pick at least two adjacent pages. We claim this would conflict with some \publicgadget{\cdot}{\frac12 n + 2}. To see this, we do some casework.
  \begin{itemize}
    \item Case 1: the pages are numbered $3\ell+1$ and $3\ell+2$ for some $\ell \in \{0, 1, ..., \deg(i)-1\}$, so they both exist during some \clausegadget{j}{\pattern{j}} where $j$ is the $\ell$th clause to contain variable $x_i$. They overlap for the duration of \publicgadget{4j-2}{\frac12 n + 2}.
    \item Case 2: the pages are numbered $3\ell+2$ and $3\ell+4$ for some $\ell \in \{0, 1, ..., \deg(i)-1\}$, so they both exist during some \clausegadget{j}{\pattern{j}} where $j$ is the $\ell$th clause to contain variable $x_i$. They overlap for the duration of \publicgadget{4j+1}{\frac12 n + 2}.
    \item Case 3: the pages are numbered $\deg(i)+1$ and $\deg(i)+2$. They both exist during \variablegadget{F} and overlap for the duration of \publicgadget{4m+2}{\frac12 n + 2}.
  \end{itemize}
  In all cases, this used up a slot of public cache during some \publicgadget{\cdot}{\frac12 n + 2}, which we have already argued is too expensive for our caching strategy. This completes the contradiction and hence $\cachingstrategy$ can achieve at most $\deg(i) + 3$ cache hits for all agents $i \in [n]$.
  
  Next, we want to argue that for each agent, it selects either all the pages congruent to zero and one mod three or all the pages congruent to zero and two mod three. We have already shown that if we consider the sorted list of only pages congruent to one and two mod three, it cannot select adjacent pages. This is already enough to deduce that it must select all pages congruent to zero mod three, some prefix of the pages congruent to one mod three (possibly empty), skip two pages, then the remaining suffix of the pages congruent to two mod three (possibly empty). It remains to show that the one and two mod three pages cannot be mixed. This is why we reduced from half-true, half-false SAT. Whenever both prefix and suffix are not empty for an agent $i$, then the caching strategy has chosen both $\page{i}{1}$ and $\page{i}{3\deg(i)+2}$. But since it has also chosen both $\page{i}{0}$ and $\page{i}{3\deg(i)+3}$, this agent $i$ will use a slot of public cache during both \publicgadget{1}{2} and \publicgadget{4m+3}{2}. But we have $n$ such agents and can only afford $n$ slots of public cache total between both of these subgadgets, so no agent can use a slot during both. Hence for every agent one of prefix or suffix must be empty, i.e. only one mod three or only two mod three pages are chosen. For agent $i$, if one mod three pages are chosen, we set $x_i$ to be true; two mod three, false. Since we have room for \publicgadget{1}{2}, we know that at most half of the variables can be true. Since we have room for \publicgadget{4m+3}{2}, we know that at most half of the variables can be false. This means exactly half are true and half are false. Since for each \clausegadget{j}{\pattern{j}}, we had enough public cache for its subgadget \publicgadget{4j-1}{\frac12 n}, we know that one of the literals in that clause is made true by this assignment choice.
  
  We have shown that $\booleanformula$ has a half-true half-false satisfying assignment if and only if there is a caching with reserves strategy for $\cachingproblem(\booleanformula)$ with at most $C$ cache misses if and only if there is a public-private caching strategy for $\cachingproblem(\booleanformula)$ with at most $C$ cache misses, as desired. This completes the proof.
\end{proof}
\end{document}